\documentclass[12pt]{article}

\pdfoutput=1 
\usepackage{jheppub} 
\usepackage[utf8]{inputenc}
\usepackage{jheppub}
\usepackage{epsfig}
\usepackage{bbm}
\usepackage{bbding}
\usepackage{url}
\usepackage{amsmath}
\usepackage{amsfonts}
\usepackage{amssymb}
\usepackage{mathtools}
\usepackage{dsfont}
\usepackage{bm}
\usepackage{graphicx}
\usepackage{longtable}
\usepackage{pdflscape}
\usepackage{subcaption}
\usepackage{color}
\usepackage{psfrag}
\usepackage[capitalise]{cleveref}
\usepackage{hyperref}
\usepackage{tikz}
\usetikzlibrary{positioning,arrows.meta}
\usepackage{pgfplots}
\usepackage{pgfplotstable}
\usepackage{subfiles}
\usepackage{longtable}
\usepackage{tabularx}
\usepackage{ltablex}
\usepackage{booktabs}
\usepackage[export]{adjustbox}
\usepackage{listings}
\usepackage{multirow} 
\usepackage{cancel,slashed}
\usepackage{graphicx}
\usepackage{latexsym}
\usepackage{transparent}
\usepackage{array}
\usepackage{makecell}
\usepackage{graphics,psfrag}
\usepackage{nowidow}
\usepackage{listings}
\usepackage[normalem]{ulem}
\usepackage{environ}
\usepackage{tikz}	
\pgfplotsset{
        compat=1.9,
        compat/bar nodes=1.8,
    }

\usepackage{amsthm}
\usepackage{array}
\usepackage{booktabs}
\usepackage{mathrsfs}
\usepackage{setspace}
\usepackage{physics}
\usepackage{simpler-wick}
\usepackage[capitalise]{cleveref}

\usepackage[most]{tcolorbox} 
\usepackage{amsmath,amsthm,amsfonts,mathtools, amssymb}

\usepackage{tabulary}                     

\usepackage{multicol}
\usepackage{pgfplots}
\pgfplotsset{compat=1.15}
\usetikzlibrary{arrows}
\definecolor{ccqqqq}{rgb}{0.8,0,0}
\definecolor{qqwuqq}{rgb}{0,0.39215686274509803,0}
\definecolor{qqqqff}{rgb}{0,0,1}
\definecolor{ffvvqq}{rgb}{1,0.3333333333333333,0}
\definecolor{cqcqcq}{rgb}{0.7529411764705882,0.7529411764705882,0.7529411764705882}

\theoremstyle{definition}
\newtheorem{theorem}{Theorem}[section]
\newtheorem{definition}[theorem]{Definition}

\newtheorem{corollary}[theorem]{Corollary}

\newcommand{\Z}{\mathbb{Z}}

\newcommand{\C}{\mathbb{C}}

\newcommand{\slz}{\text{SL}(2,\mathbb Z)}

\renewcommand{\t}{(\tau)}

\newcommand{\spg}{\text{Sp}(2g,\mathbb Z)}

\newcommand{\spz}{\text{Sp}(4,\mathbb Z)}

\renewcommand{\t}[0]{(\tau)}

\newcommand{\coma}{\, , \quad}
\newcommand{\fstop}{\, .}
\newcommand{\sef}{\mathcal{E}_4}
\newcommand{\ses}{\mathcal{E}_6}
\newcommand{\sew}[1]{\mathcal{#1}}

\newcommand{\mc}{\mathcal}
\newcommand{\h}{\mathbb{H}}

\def\ord{\text{ord}}

\def\im{{\text{Im} \,}}

\definecolor{AWG}{RGB}{212,175,55} 
\definecolor{JMLB}{RGB}{4, 55, 242} 

\newcommand{\env}[2]{\begin{#1} #2 \end{#1}}

\graphicspath{{Figures/}}

\begin{document}

	\pagestyle{plain}

	\makeatletter
	\@addtoreset{equation}{section}
	\makeatother
	\renewcommand{\theequation}{\thesection.\arabic{equation}}
	\pagestyle{empty}

\rightline{DESY-26-063}
\vspace{2cm}

\begin{center}
	\LARGE{\bf Automorphic Structures of \\ Heterotic Vacua\vspace{4mm}}\\
	\large{
    Jacob M. Leedom,\textsuperscript{\emph{a}} Nicole Righi,\textsuperscript{\emph{b},\emph{c}}  Alexander Westphal\textsuperscript{\emph{d}}\\[4mm]}
	\footnotesize{
    \textsuperscript{\emph{a}}CEICO, Institute of Physics of the Czech Academy of Sciences\\
Na Slovance 2, 182 00 Prague 8, Czech Republic\\
	\textsuperscript{\emph{b}}Scuola Normale Superiore, Piazza dei Cavalieri 7, 56126 Pisa, Italy\\
    \textsuperscript{\emph{c}}INFN, Sezione di Pisa, Largo Bruno Pontecorvo 3, 56127 Pisa, Italy\\
    \textsuperscript{\emph{d}}Deutsches Elektronen-Synchrotron DESY, Notkestr. 85, 22607 Hamburg, Germany}

\vspace*{10mm}

\small{\bf Abstract} 
\\[4mm]
\end{center}
\begin{center}
\begin{minipage}[h]{\textwidth}
We study moduli stabilization in 4D effective field theories with Sp(4,$\mathbb{Z}$) self-duality inspired by heterotic orbifold compactifications with Wilson lines. The target-space duality group of these theories is enhanced from SL$(2,\mathbb{Z})$ to Sp$(4,\mathbb{Z})$, making Siegel modular forms the appropriate language to formulate the effective supergravity action. We construct the corresponding effective theory including moduli-dependent threshold corrections to the gauge kinetic function and nonperturbative effects in the superpotential. The degeneration limit of the Wilson lines distinguishes different sectors and dictates which combination of cusp forms appears in threshold corrections. We compute the resulting scalar potential and prove several general statements about its extrema. In particular, we show that the fixed points of Sp$(4,\mathbb{Z})$ are extrema of the potential, and derive genus-2 analogues of no-go theorems for de Sitter vacua. Finally, we show how positive-energy metastable minima can arise once supersymmetry is broken in the dilaton direction by nonperturbative contributions to the K\"ahler potential.

\end{minipage}
\end{center}

	\newpage
	\setcounter{page}{1}
	\pagestyle{plain}
	\renewcommand{\thefootnote}{\arabic{footnote}}
	\setcounter{footnote}{0}
	
	\tableofcontents

\setcounter{page}{1}

\section{Introduction}
\label{sec:intro}
Dualities play an essential role in the modern understanding of quantum field theory and string theory and their underlying mathematical structure. Of particular use are \textit{self-dualities}, which map different regions of a parameter space to each other. Moreover, such dualities are increasingly found to constrain the scalar potential and the structure and distribution of the resulting vacua (i.e.~local minima of the scalar moduli potential) of string compactifications as well. The quintessential example is S-duality in $10$D type IIB string theory; 
it restricts the effective action of the theory and implies that the type IIB pseudo-action must be a \textit{modular function}, i.e.~it must be invariant under modular transformations.
This modular invariance manifests itself by restricting the functions that can appear in the effective action: they must themselves be modular functions, or at least modular forms with well-defined transformation properties under $\slz$. A relevant example of such functions is the real-analytic Eisenstein series that shows up as the coefficient of the $R^4$ correction to the type IIB effective action~\cite{Green:1997tv,Green:2016tfs}. 

It is natural to wonder if similarly powerful statements can exist in compactifications of string theory and be relevant for applications of string theory to particle phenomenology and cosmology. In the context of toroidal orbifold compactifications of heterotic string theory, the relevant self-duality is T-duality, which acts on geometric moduli as opposed to the string coupling. 
A heterotic compactification on $T^6$ will be associated with an even, self-dual lattice, the Narain lattice~\cite{Narain:1985jj,Narain:1986am}. For a $6$D compactification manifold, this is a lattice of signature $(6,22).$ 
The T-duality group then corresponds to the outer automorphism group of this lattice.  This defines the geometric moduli space of the compactification as
\begin{equation}
    SO(2,2)/SO(2)\times SO(2)\,.
\end{equation}
As a self-duality of the theory, T-duality must manifest as a target space symmetry of the effective field theory (EFT). This directly implies that the action of the EFT must be invariant under the action of any element of
$SO(2,2,\mathbb{Z})\cong \slz$ -- that is, the action must be an \emph{automorphic function}. This point has been emphasized and utilized in a number of works~\cite{Ferrara:1989bc,Font:1990nt,Cvetic:1991qm,Gonzalo:2018guu,Leedom:2022zdm,Cribiori:2023sch}. The simplest example is an $\slz$ subgroup acting on an overall K\"{a}hler modulus $T$ as
\begin{equation}
    T\rightarrow \gamma\cdot T = \frac{aT+b}{cT+d}\,.
\end{equation}
Then, enforcing this transformation to be a symmetry of the effective action allows one to restrict the types of functions that can appear in the effective theory. 

The symmetry group is enhanced once one takes into account the presence of Wilson lines. The classical moduli space as well and the classical T-duality transformations for heterotic string compactifications with $p$ Wilson lines is given by the coset
\begin{align}
    SO(2+p,2)/SO(2+p)\times SO(2)\,,
\end{align} and the T-duality group is given by $SO(2+p, 2, \mathbb{Z})$.
 The simplest extension to theories with $\slz$ modular group is then considering one Wilson line, so that the duality group is enhanced to $SO(3,2,\mathbb{Z})\cong \spz$. The structure of this larger duality group also implies that besides the overall K\"ahler modulus $T$ and the Wilson line $Z$ we need to include the relevant bulk complex structure modulus $U$ as well in order to parametrize the so-called \emph{Siegel upper half space of degree 2}: this is the manifold spanned by the three overall moduli $T,U,Z$ of a toroidal heterotic string compactification, of which $\spz$ is the modular symmetry.
 
 The general framework for $\mathcal{N}=1$ supersymmetric theories with $\spz$ modular symmetries was constructed in the 90's by~\cite{LopesCardoso:1994is,LopesCardoso:1996nc,Stieberger:1998yi,Nilles:1997vk} and revived recently in~\cite{Ding:2020zxw,Baur:2020yjl,Nilles:2021glx} in the context of flavor symmetries (see also \cite{Ding:2021iqp,Ding:2023htn,Ding:2024xhz,Funakoshi:2024yxg,Jiang:2025qbi}). Most importantly, the familiar set of classical modular forms is no longer appropriate, and one must employ \emph{Siegel modular forms of degree 2}. 
 
In this work, we consider for the first time moduli stabilization in the context of $\spz$ modular symmetry. In \cref{sec:tororbifSiegel}, after reviewing the $TUZ$-setup in a modern fashion, we will move to the full $STUZ$ model which includes as well the dilaton contained in the chiral multiplet $S$ as the one further scalar field universally present in every string theory. The dilaton in turn sets the string coupling and within heterotic string theory also controls the gauge coupling of the heterotic gauge sector. We then discuss how to write the nonperturbative superpotential, including the appearance of moduli-dependent threshold corrections to the effective gauge coupling~\cite{Dixon:1990pc,Mayr:1995rx,Stieberger:1998yi}. Moreover, we construct a modular-invariant function to parametrize nonperturbative corrections to the superpotential in the geometric moduli.  Furthermore, we  introduce nonperturbative corrections to the K\"ahler potential arising from purely stringy instanton contributions, as in~\cite{Shenker:1990,Silverstein:1996xp,Gaillard:2007jr,Kaufman:2013pya,Leedom:2022zdm,Green:2016tfs,Alvarez-Garcia:2024vnr}. In~\cref{sec:scalarpot}, we explicitly compute the F-term scalar potential and in~\cref{sec:nogos} we prove a number of theorems showing that the elliptic fixed points of $\spz$ are extrema of the potential. Finally, in~\cref{sec:minima}, we focus on characterizing both the elliptic fixed points as well as points inside the moduli space: we discuss the conditions for them to be minima of the potential and how one could potentially realize a vacuum with positive energy. To make the discussion explicit and evaluate numerically our findings, we created the Mathematica package~\hyperlink{https://github.com/nrighi/smfs.git}{SMFs.m} available at~\cite{righi:2026SMF}. 

As for mathematical tools, we will make use of the theory of Siegel modular forms (SMFs), reviewed in~\cref{app:SMFs};  we will discuss the modular invariant functions built from the ring, the geometry of the Siegel upper-half space and its zeroes and poles,  and how they affect the EFT. Moreover, we prove new mathematical results and introduce the notion of \emph{bipartite} Siegel modular forms.

\section{Toroidal orbifolds with Siegel modular symmetries}
\label{sec:tororbifSiegel}
Since the $T^6$ compactification of heterotic string theory results in an EFT with too many preserved supercharges, we want to consider toroidal orbifolds, i.e.~compactifications on a $6$D manifold $X_6$ of the form $
    X_6 = T^6/\Gamma$ ,
where $\Gamma$ is a finite cyclic group of the form $\mathbb{Z}_M$ or $\mathbb{Z}_M\times \mathbb{Z}_N$. The T-duality group is then a subgroup that preserves the action of $\Gamma$. In this work, we focus on $X_6$ with the further restricted form $T^6/\Gamma = T^4/\Gamma \times T^2$ or $T^6/\Gamma = K3/\Gamma \times T^2$. Such compactifications preserve an $\mathcal{N}=2$ supersymmetric subsector on the torus not touched by the orbifolding action. The importance of this is that the nonperturbative superpotential induced from gaugino condensation will depend non-trivially on the moduli of the untouched $T^2$. Given the discussion in the introduction, this non-trivial dependence must take the form of a particular class of modular forms. This situation has been studied in various contexts in the case of vanishing Wilson line moduli, where the relevant T-duality group is $\slz_T \times \slz_U \rtimes \mathbb{Z}_2$. In what follows, we  study how the insertion of a Wilson line changes the story.

\subsection{4D $\mathcal{N}=1$ heterotic supergravity}
Let us consider the heterotic geometric moduli associated with a $T^2$ factor in the compactification manifold $X_6$. We have a K\"{a}hler modulus $T$, a complex structure modulus $U$, and Wilson line $Z$. In terms of the metric $G_{ab}$, 2-form $B_2$, and gauge field $A$, the complex moduli are defined as~\cite{Mayr:1995rx,LopesCardoso:1994is}
\begin{align}
    T &= B_{12}+i\sqrt{\text{det}(G)} +A_1 (-A_2+U A_1)\coma\\
    U &= \frac{1}{G_{11}}\bigg(G_{12}+i\sqrt{\text{det}(G)}\bigg)\coma\\
    Z &= -A_2 + U A_1\fstop
\end{align}
It is convenient to
combine these fields into a $2\times 2$ matrix belonging to the Siegel upper half-space of genus two $\mathbb{H}_2$ as
\begin{equation}
M=
\begin{pmatrix}
    T & Z\\
    Z & U
\end{pmatrix}\,,\quad \im M > 0\,,
\end{equation}
such that 
\begin{equation}
    Y:=-\frac{1}{4}\text{det}(M-M^{\dag})=-\frac{1}{4}\left((T-\bar{T})(U-\bar{U})-(Z-\bar{Z})^2\right)\!=\im{T}\,\im{U}-\im{Z}^2
\end{equation}
is positive definite. Hence, the K\"ahler potential for the fields $T,\,U,\,Z$ is
\begin{equation}
    K^{(0)}= -\gamma\ln Y\,, \quad \gamma\in\mathbb{R}\,.
\label{eq:siegel_kahler_pot}
\end{equation}
Note that, for vanishing Wilson line, the K\"ahler potential reduces to the sum of $K(T,\bar T)$ and $K(U,\bar U)$. We will frequently make use of the  $Z\rightarrow 0$ limit as a consistency check with known results from the $\slz$-symmetric results.

Under an $\spz$ transformation
the K\"ahler potential undergoes a K\"ahler transformation
\begin{equation}\label{eq:Ktransformation}
\begin{split}
    K^{(0)} &\rightarrow -\gamma\ln\left[\text{det}\left((CM^\dag +D)^{-1}(M-M^\dag)(CM +D)^{-1}\right)\right]\\
        &\rightarrow K^{(0)} +\gamma\ln\left[\text{det}(CM^\dag +D)\right] +\gamma\ln\left[\text{det}(CM +D)\right]\,.
\end{split}
\end{equation}
This allows us to anticipate the required automorphic properties of the nonperturbative superpotential. Since the total F-term scalar potential 
\begin{equation}\label{eq:Fpotential}
    V=e^{G} \left(G_i G^{i\bar j}G_{\bar j}-3\right)= e^{K}\left(K^{i\bar{j}}\mathcal{D}_iW\mathcal{D}_{\bar{j}}\overline{W}-3|W|^2\right)
\end{equation}
must be a modular invariant function, so must be the auxiliary supergravity field $G\equiv K+\ln|W|^2$. Namely, $W$ must transform to offset the K\"{a}hler transformation in \cref{eq:Ktransformation} and restore modular invariance. In particular, we need
\begin{equation}
        W \rightarrow e^{i\delta}\text{det}(CM+D)^{-\gamma}\,W\,,
\end{equation}
where $\delta$ is a phase that becomes relevant only in the case of a non-trivial multiplier system. The above transformation property provides a constraint on the form and the possible component terms which the superpotential may contain. We then need physical arguments and sources to determine the actual contributions to the superpotential. It turns out that gaugino condensation in the dilaton sector exactly reproduces the modular behavior we are looking for, once moduli-dependent threshold corrections are included. This is the subject of the next section. 

Before moving on, let us remark that the presence of a non-trivial Wilson line forces us to move from the familiar set of elliptic modular forms to Siegel modular forms (SMFs) of degree 2. We introduce here the basics which will be useful to understand the EFT; We defer the reader to \cref{app:SMFs} for details on SMFs. 

First, the notion of the hyperbolic upper-half plane of genus $1$ generalizes to the \textit{Siegel upper-half space} of genus $2$:
$\displaystyle \mathbb{H}_2 := \left \{M \in \text{Mat}(2\times 2,\C) \middle\vert M^T = M, \ \Im(M) >0 \right \}$. This is an open subset in the space of 
$2\times 2$ complex symmetric matrices, hence it is a complex manifold of complex dimension 
$3$. The Siegel modular group of degree 2 is $\spz$. Then, the symplectic group $\spz$ acts on the period matrix $M$ as 
\env{align}{
M \rightarrow \Gamma_2 M = (AM + B)(CM + D)
^{-1}, \ \ \ \forall\; \Gamma_2 = \env{pmatrix}{A & B \\ C & D} \in \spz\,.}
 A holomorphic function $F_k: \h_2 \rightarrow \C$ that transforms  under $\spz$ as
\env{align}{F_k(M) \rightarrow F_k(\Gamma_2 M) = (CM + D)^k F_k(M)\,, \quad \Gamma_2 \in \spz\,,}
is a holomorphic \textit{Siegel modular form of degree $2$ and weight $k$}. An elliptic modular form can be thought of as a Siegel modular form of degree $1$ by virtue of the isomorphism that $\spg \cong SO(g+1,g)$ and that $SO(2,1) \cong \slz$. To conclude this quick introduction, we remark a classical theorem by Igusa \cite{Igusa1962}, stating that the ring of algebraically independent SMFs of degree 2  over $\mathbb{C}$ is generated by the Siegel-Eisenstein forms $\sef$, $\ses$, and the cusp forms $\chi_{10}$ and $\chi_{12}$.

\subsection{Moduli-dependent threshold corrections}
Before discussing the superpotential, let us include in the K\"ahler potential the dependence on the dilaton $S$ as follows
\begin{equation}\label{eq:kpot}
    K=-\ln\left(S+\bar{S}+\frac{1}{8\pi^2}K^{(1)}(M)+k_{np}(S)\right)+K^{(0)}(M)\,,
\end{equation}
where $K^{(1)}$ denotes the moduli-dependent one-loop correction to the tree-level part $K^{(0)}$ of the K\"ahler potential. It contains the Green-Schwarz  anomaly cancellation corrections $K^{(1)}_{\mc{N}=1}$, which refer to twisted planes and therefore involve only the moduli of these planes, and Green-Schwarz corrections $K^{(1)}_{\mc{N}=2}$ from the untwisted planes, which can be computed from the $\mc{N}=2$ prepotential. We also include nonperturbative corrections to the dilaton $k_{np}(S)$, which will be relevant in~\cref{sec:minima}. Below, we will refer to the first parenthesis in \cref{eq:kpot} as $k(S,\bar S)$.

In orbifold compactifications with untwisted planes, threshold corrections are given in terms of automorphic functions of the underlying target space duality group \emph{with the appropriate singularity structure}. In 4D $\mc{N} = 2 $ heterotic string compactifications,  threshold corrections can be written in terms of the supersymmetric index, which is related to the perturbative heterotic $\mc{N} = 2$ prepotential. The Abelian gauge threshold functions are related to the second derivatives of the one-loop prepotential. At the loci of enhanced non-Abelian gauge symmetries, some of the Abelian gauge couplings will exhibit logarithmic singularities due to the additional massless states entering in the loop. We now explain more carefully the above statements.

The evolution equation of the effective coupling below the string scale $M_s$ is given by~\cite{Kaplunovsky:1995jw,Nilles:1997vk}
\begin{equation}\label{eq:running}
    \frac{1}{g_{a}^{2}(\Lambda^2)}=k_a \left(\Im S+\frac{\Delta_{univ}}{16\pi^2}\right)+\frac{b_a}{16\pi^2}\ln\frac{M_{s}^2}{\Lambda^2}+\frac{1}{16\pi^2}\Delta_a\,,
\end{equation}
where $k_a$ is the level of the Kac-Moody algebra associated to the gauge group labelled by $a$, $b_a$ are related to the one-loop $\beta$-function coefficients of the  $\mc{N}=2$ subsector,  the group-dependent threshold corrections are
contained in $\Delta_a$, and there is also a universal correction $\Delta_{univ}$ which contains $K^{(1)}$ and can be absorbed
in the holomorphic redefinition of the dilaton. Both $\Delta_a$ and $\Delta_{univ}$ depend on the untwisted moduli and include contribution of
massive string and Kaluza-Klein states. In particular, for any perturbative heterotic compactification we have~\cite{Kaplunovsky:1987rp,Dixon:1990pc,Antoniadis:1991fh,Antoniadis:1992sa,Antoniadis:1992rq,Kaplunovsky:1995jw}
\begin{equation}\label{eq:intthresh}
    \Delta_a =\int_\mc{F} \frac{d^2\tau}{\tau_2} \Big[\mathcal{B}_a(\tau,\bar{\tau})-b_a\Big]\,,
\end{equation} 
where $\mc{F}$ is the fundamental domain for the worldsheet torus parameter $\tau$. $\mathcal{B}_a$ is given by the trace over  the Ramond sector of the worldsheet CFT
\begin{equation}
    \mathcal{B}_a(\tau,\bar\tau)=-\frac{1}{|\eta(\tau)|^2}\text{Tr}_{RR}\left\{e^{J_0}(-1)^{\bar F}(-1)^F q^{L_0-c/24}\bar{q}^{\bar{L}_0-\bar{c}/24}\left(Q_a^2-\frac{k_a}{8 \pi \tau_2}\right)\right\}\,,
\end{equation}
where $Q_a$ is a generator of the gauge group and $J_0$ is the U$(1)$ current of the left $\mc{N}=2$ Kac-Moody algebra. This quantity can be written in terms of the supersymmetric index \cite{Antoniadis:1992rq,Harvey:1995fq}.
Since the supersymmetric index is a weight 0 index 1 weak Jacobi form, $\Delta_a$ can be entirely written in  terms of appropriate modular forms. Therefore, it is $\Delta_a$ that gives rise to the automorphic properties encoding the  corrections to the gauge coupling, which depend \emph{solely} on the geometry of the internal manifold. Moreover, $\Delta_a$ contains the contribution of all the massive modes, exemplifying that heterotic orbifold models form a beautiful example of string vacua for which the entire massive spectrum is known exactly. Additionally, $\Delta_a$ can contain a finite number of states which become massless at some critical point $P_c$. There, it develops a logarithmic singularity of the form $n\log(M_i- P_c)$, where $n$ accounts for the degeneracy of states becoming massless and is proportional to the order of $P_c$ under modular transformation~\cite{LopesCardoso:1994ik}. This was already noticed in classic literature~\cite{Binetruy:1996uv,deCarlos:1992kox} about heterotic gaugino condensation: integrate out the massive quark fields appearing in the condensate one-loop determinant, which have their masses set by moduli singlet scalars. The result of integrating out the quarks is to render the one-loop determinant a power-law function of the modulus setting the  quark mass. This modulus dependence of the one-loop determinant in turn can be written as a singular logarithmic contribution to the one-loop threshold corrections of the gauge kinetic function in heterotic string theory.

For 4D orbifold models with $\slz$ modular symmetries, it was shown~\cite{Dixon:1990pc,Harvey:1995fq} that $\mc{B}_a$ contains combinations of modular forms such that performing the integral in \cref{eq:intthresh} gives one term proportional to $T$ and, most importantly, one term proportional to
\begin{equation}
    e^{\pi i T/12}\prod_{n>0} \left(1-e^{2\pi n i T}\right) \equiv \eta(T)\,,
\end{equation}
generated by worldsheet instanton contributions. Hence, for a single $T^2$ one gets
\begin{equation}
    \Delta_a=- b_a\ln\left[(-iT+i\bar T)|\eta(T)|^4|\right]\,.
\end{equation}
To completely cancel the T-duality anomaly, it is crucial to write the gauge kinetic function $f_a$ including the above threshold corrections. Finally, note that the effective coupling does not receive higher-loop perturbative corrections, as they would violate the  shift invariance of the dilaton:~\cref{eq:running} is exact in perturbation theory.

An argument by close analogy was given for theories with Wilson lines, which are then symmetric under the $\spz$ modular group. First,~\cite{Mayr:1995rx} extracted $\Delta_a$ by looking at the perturbative duality symmetry $SO(3, 2, \mathbb{Z})$ and at  the singularity structure in the moduli space. Mayr and Stieberger identified two relevant forms of behavior for the threshold corrections.  In the first case, no (under  the considered gauge group) charged particles become massless for $ Z\rightarrow 0$; the form of these thresholds is given by
\begin{equation}\label{eq:deltaII}
    \Delta_a^{II}\sim-\frac{1}{12}\ln(Y^{12}|\chi_{12}|^2)\,.
\end{equation}
In the second case, some particles charged under the gauge group under consideration become massless for $Z \rightarrow 0$, and the threshold corrections are given by
\begin{equation}\label{eq:deltaI}
     \Delta_a^I\sim-\frac{1}{10}\ln(Y^{10}|\chi_{10}|^2)\,.
\end{equation}
Finally, in~\cite{Stieberger:1998yi} Stieberger also  considers the case of two gauge groups, where one gets enhanced for $ Z\rightarrow 0$, say $a$, and the other, labelled by $a^\prime$,  does not:
\begin{equation}\label{eq:deltaaap}
     \Delta_{a-a'}\sim-\frac{1}{12}\ln(Y^{12}|\chi_{12}|^2)+ \left(1-\frac{b_a(Z=0)}{b_{a^\prime}(Z\neq 0)}\right)\ln\left|\frac{\chi_{10}^{1/2}}{\chi_{12}^{5/12}}\right|^2\,.
\end{equation}
$\Delta_a^I$ was first explicitly derived by computing one-loop string diagrams on the world-sheet on a standard embedding of the heterotic string compactified on K3$ \times T^2$ with a Wilson line turned on~\cite{Kawai:1995hy}, where one is helped by the presence of a K3 surface, where the elliptic genus, and hence the supersymmetric index, is well know.  On general grounds,~\cref{eq:intthresh} can be recasted into three integrals over different orbits of $\spz$ (just like in the $\slz$ case) such that one finds~\cite{Kawai:1995hy,Harvey:1995fq,Neumann:1996is}
\begin{equation}
    \frac{\Delta_a}{b_a}=-\ln(Y^{c(0,0)/2})-\ln\Big|e^{AT+CU+BZ}\prod_{(s,t,u)>0}\left(1-e^{sT+tU+uZ}\right)^{c(st,u)}\Big|^2\,.
\end{equation}
For details, we refer the reader to \cref{app:SMFs}. Gritsenko and Nikulin~\cite{gritsenko1995siegel} proved that the above relation contains the exponential lift of a weak Jacobi form~\cite{borcherds1995automorphic}. The result of the exponential lift is a genus-$2$ Siegel modular form for a (possibly congruence) subgroup of $\spz$ with weight $c(0,0)/2$, $c(0,0)$ being the constant term of the weak Jacobi form of weight $0$ and index $1$ that has been exponentiated. Note that the factor $Y$ transforms with weight $-1$, which fixes the weight of the second piece to be $c(0,0)/2$.

Applied to  K3$ \times T^2$, we have that $c(0,0)=20$, 
and 
\begin{equation}
e^{T+U+Z}\prod_{(s,t,u)>0}\left(1-e^{sT+tU+uZ}\right)^{c(st,u)}\equiv\chi_{10}(M)\,,
\end{equation}
such that the threshold correction is a modular function of weight $0$ and  one reproduces~\cref{eq:deltaI}.

However, it seems that $\Delta_a^{II}$ cannot be derived from~\cref{eq:intthresh}. Indeed,~\cite{Stieberger:1998yi} showed that differences of gauge threshold corrections $\Delta=(\Delta_a-\Delta_{a'})/(b_a-b_{a'})$ of any K3 orbifold are always of the form $\Delta_a^I$ in \cref{eq:deltaI} or $\Delta_{a-a'}$ in \cref{eq:deltaaap} (note that all the model dependence goes into the prefactor $b_a$). Those can both be derived from the exponential lift of a weight zero weak Jacobi form. The problem about $\Delta_a^{II}$ lies in the properties of $\chi_{12}$, namely its zeroes and poles structure is not given by a divisor with a simple form~\cite{Stieberger:1998yi}.\footnote{By a simple divisor we mean a union of dimension-2 loci where particular symmetries are enhanced (the higher-dimensional analogues of the fixed points). Being defined on simple divisors is a fingerprints of exponential lifts. However, the divisor of $\chi_{12}$ seems to be of a different, non-trivial type.}

Let us now recap the physical argument for appearance of both $\chi_{10}$ and $\chi_{12}$ in the one-loop threshold corrections: Some sectors of a given compactification show a finite number of states becoming massless in the limit of vanishing Wilson line $Z\to 0$, while other sectors do not have such light states at the locus $Z=0$. The threshold corrections must be parametrized by logs  of $\spz$-invariant cusp forms. Hence, for the sectors with light states at $Z=0$ the threshold contribution must take the form $\ln(\text{cusp-form}(Z))$ with a cusp form vanishing as a power of $Z$ at $Z=0$. Conversely, for sectors with no light states at the zero Wilson line locus, their contribution to the threshold corrections should look like $\ln(\text{cusp-form})$ with a non-vanishing cusp form at $Z=0$. The only two $\spz$-invariant cusp forms available are $\chi_{10}$ and $\chi_{12}$, respectively. Hence, by these physics reasons we expect the threshold correction for a compactification to have contributions of both types, one of each type, which qualitatively resemble \eqref{eq:deltaaap}.

For our purposes, the moduli-dependent threshold corrections allows us to write the superpotential $W$ as
\begin{equation}\label{eq:superpot}
    W(S,M)=\frac{\Omega(S)}{\chi_\beta^{\alpha}(M)}\equiv\frac{\Omega(S)}{\chi_\gamma(M)}\,,
\end{equation}
where the piece proportional to $Y$ in the threshold correction has been absorbed in the definition of $S$. The function $\Omega(S)$ encodes only the $S$-dependence, as e.g.~$\Omega(S)\sim h+ e^{-S/b_a}$ for a single condensate and $h$ an additive constant to describe $H_3$-flux effects. In the denominator we have included the contribution from the threshold corrections as weight-$\beta$ cusp-forms to the power $\alpha$; for $\gamma=\alpha\beta$, modular invariance of the supergravity field $G$ is ensured. For this reason we introduce the shorthand notation $\chi_\beta^{\alpha}\equiv\chi_\gamma$. Below we keep the discussion general, but it will be important to pick the right threshold correction when computing an explicit example.

\subsection{Nonperturbative corrections from geometric moduli}
Another source of contribution to the superpotential comes from nonperturbative effects in the $M=T,\,U,\,Z$ moduli sector. This should be encoded in a modular function $\mc H(M)$, since threshold corrections are already saturating the modular weights for the superpotential. For the sake of full generality, such function better be the most general modular function one can write. Moreover, it is expected to be the genus-$2$ generalization of the following modular function: 
\begin{align}
\label{eq:H}
    H(T) = (j(T)-1728)^{m/2}\,j(T)^{n/3}\,\mathcal P(j(T)) =\left ( \frac{E_6(T)}{\eta (T)^{12}}\right)^m \left( \frac{E_4(T)}{\eta(T)^{8}}\right)^n \mathcal P(j(T))\,,
\end{align}
where $\mathcal P(j\t)$ is a polynomial of the $j$-function of finite order and we used the fact that, at the elliptic fixed points, $j(i)=1728$ and $j(\rho\equiv e^{2\pi i/3})=0$, and\begin{equation}\label{eq:reljE}
    j(T)=\frac{E_4(T)^3 }{\eta(T)^{24}}\,,\qquad j(T)-j(i)=\frac{E_6(T)^2 }{\eta(T)^{24}}\,.
\end{equation}
Note that $H(T)$ have the schematic form of $\exp(\pm2\pi iT)$: that is why the function
$H(T) $ can be thought of parametrizing nonperturbative effects in the K\"ahler modulus. Note also that $m$ and $n$ are divided by their respective stabilizer orders.\footnote{On the quotient $\slz\backslash \mathbb H^*$, a holomorphic,  invariant function must be analytic in the orbifold local coordinates
$u=(\tau-i)^2$, $v=(\tau-\rho)^3$,
because the identifications act as $\tau-i\mapsto-(\tau-i)$ at $i$, and $\tau-\rho\mapsto \zeta_3(\tau-\rho)$ at $\rho$ ($\zeta_3$ being the multiplier). 
So a genuine modular function cannot have an odd order zero in $(\tau-i)$, nor an order that is not a multiple of $3$ in $(\tau-\rho)$.}

The expression in~\cref{eq:H} is the most generic regular modular function on the  $\slz$ upper-half plane (with a potentially non-trivial multiplier system)~\cite{Rademacher,Lehner:1964}. This follows from the classical result that every rational function of the $j$-invariant function is a modular function of $\slz$. Conversely, every modular function of $\slz$ of weight 0 is a rational function of the $j$-invariant. In the following, we propose a possible genus-$2$ generalization of the modular function $H(T)$.

First, in the genus-$2$ setup, there is not a single $j$-invariant as for genus 1, where a single invariant function is enough to  classify generic elliptic curves. In particular, let $\Gamma$ be any subgroup of $\mathbb{H}_g$. 
Then there exist $s = g(g + 1)/2$  analytically and algebraically independent automorphic functions with respect to $\Gamma$ which are quotients of automorphic forms \cite{Klingen_1990}[Prop.~2]. The genus-$2$ moduli space $\mc A_2$ has complex dimension 3, so one needs $s=3$ algebraically independent modular functions to parametrize it. The genus-$2$ analogue of the $j$-invariant is then a triplet of modular invariant functions that generates the field of Siegel modular functions for $\spz$.

While all the  results stated below can be generalized to degree $g$, let us specify the discussion to the case of interest. We follow~\cite{Klingen_1990}.

Let $Q_2$ be the following  algebraic function field of meromorphic
functions on $\mathbb{H}_2$:
\begin{equation}
    Q_2=\left\{ f=\frac{g}{h}\,\big|\,g,h \text{ modular form of equal weight, } h\neq 0\right\}\,.
\end{equation}
One can prove that there exist $3$ algebraically independent
modular functions in $Q_2$. Hence, any modular function $f$ in $Q_2$ can be represented \emph{globally} as the quotient of isobaric polynomials in $\sef,\, \ses,\, \mathcal{E}_{10}$ and $\mathcal{E}_{12}$ of equal weight. In particular, the field $Q_2$ is the rational function field generated over $\mathbb{C}$ by the algebraically independent functions~\cite{freitag2013siegelsche}[Thm.~1.7]
\begin{equation}\label{eq:deg2j}
    \frac{\sef^5}{\mathcal{E}_{10}^2} \,,\quad
    \frac{\ses^2}{\mathcal{E}_{12}} \,,\quad 
    \frac{\sef \ses}{\mathcal{E}_{10}}\,  .
\end{equation}
 This statement generalizes the well-known result on the generation of the field of elliptic modular functions by the modular invariant $j(\tau)$, therefore we will denote the functions in~\cref{eq:deg2j} as $\text{J}_1(M),\, \text{J}_2(M),\,\text{J}_3(M)$, respectively.

We pause here to examine whether the proposal above is compatible with the geometry of the moduli space. Since the quotient $\mathbb H_g/\Gamma_g$ is non-compact, a priori a $\Gamma_g$-invariant meromorphic function on $\mathbb H_g$ need not extend to a meromorphic function on the compactification: near the cusps it could develop non-meromorphic behavior, for instance an essential singularity, and therefore fail to belong to the classical field of modular functions.

For $g=1$, an $\slz$-invariant meromorphic function on the upper half-plane $\mathbb H_1$ is a classical modular function only if one also imposes meromorphicity at the cusp, i.e.~at $q=e^{2\pi i \tau}=0$ as $\tau\to i\infty$. Without this extra condition, one allows infinitely many $\slz$-invariant functions with essential singularities at $q=0$. These are invariant functions on  $\mathbb{H}_1$ but they do not define modular functions in the usual algebraic sense. As a consequence, the resulting field looses algebraic control.

For genera $g>1$, however, this extra condition is not needed. Baily proved in~\cite{Baily}[Thm.~5] that the meromorphicity conditions at the cusps is guaranteed by the non-trivial geometry. The reason is that the case $g=1$ is geometrically exceptional. The modular curve $\mathcal A_1=\mathbb H_1/\Gamma_1$ has a boundary of codimension 1, so the cusp behaves like a genuine puncture on a  curve, and poles or more general singularities may occur there. By contrast, for $g>1$ the boundary of the compactification  $\mathcal A_g=\mathbb H_g/\Gamma_g$ has codimension at least 2. The boundary is therefore 
much smaller from the point of view of several complex variables: it is too small to support independent pole data for holomorphic automorphic objects. Equivalently, the extra complex directions available near the cusps rigidify the behavior of invariant functions and prevent the kind of arbitrary singular behavior that may occur in genus 1. This is the geometric reason why, in higher genus, no separate growth condition at the cusps is required: $\Gamma_g$-invariance together with meromorphicity on $\mathbb H_g$ already forces the correct boundary behavior. The principle behind these statements is the \emph{pseudoconcavity} property of the moduli space $\mc A_g$ for $g>1$~\cite{Andreotti61,Klingen_1990}: even if $\mc A_g$ is non-compact, its geometry at infinity is sufficiently concave to impose strong finiteness properties on meromorphic functions.

Given these premises, we are ready to formulate the following
\begin{definition}
Given the 6 elliptic fixed points $\sigma_i$, $i=1,\dots,6$ with stabilizer subgroups Stab$(\sigma_i)$ of order $\ord(\sigma_i)$, we propose a candidate for a suitably general modular invariant function in the genus-$2$ Siegel upper-half space, which is given by
\begin{equation}
\mc{H}(M)=
\prod_{\substack{1\le i\le 6}}
\left(\text{J}_1-\text{J}_1|_{\sigma_i}\right)^{\frac{m_i}{\ord(\sigma_i)}}
\left(\text{J}_2-\text{J}_2|_{\sigma_i}\right)^{\frac{n_i}{\ord(\sigma_i)}}
\left(\text{J}_3-\text{J}_3|_{\sigma_i}\right)^{\frac{\ell_i}{\ord(\sigma_i)}}
\mathcal{P}(\text{J}_1,\text{J}_2,\text{J}_3)\,.
\end{equation}
Suitably general here means that the choice of $\mc{H}$ avoids the appearance of inverse powers of the Siegel cusp forms $\chi_{10}$ and $\chi_{12}$ such that $\mc{H}$ has no poles  at the fix points.
\end{definition}
\noindent For genus 1 elliptic modular forms, in defining the most general modular-invariant function one has to worry about the presence of multiplier systems. However, as reviewed in~\cite{Klingen_1990}, it was proven~\cite{Christian:1962uc} that for genera $n>1$ only integral weights are possible, implying that for $n = 2$ only one non-trivial character
exists, and none for $n > 2$. In~\cite{maass1964multiplikatorsysteme} (see also~\cite{freitag2024multiplier}) it is furthermore shown that, for the case $n=2$ of interest here, the two possible multiplier systems come down to either $+1$ or $-1$.

We can evaluate $\mc H(M)$ numerically at the fixed points and rearrange it as follows:
\begin{equation}
        \mc{H}(M)\!= \!\left(\!\frac{\sef^5}{\mathcal{E}_{10}^2}\!-\!14.1687\!\right)^{\!\!\frac{m_5}{24}} \!\!\left(\!\frac{\ses^2}{\mc{E}_{12}}\!-\!14.7414\!\right)^{\!\!\frac{n_2}{48}} \!\!\left(\!\frac{\sef^5}{\mc{E}_{10}^2}\!\right)^{\!\!\mu} \!\left(\!\frac{\ses^2}{\mc{E}_{12}}\!\right)^{\!\!\nu}\! \left(\!\frac{\sef \ses}{\mc{E}_{10}}\!\right)^{\!\!\lambda}\mathcal{P}(\text{J}_1,\text{J}_2,\text{J}_3)\,,
\end{equation}
where we have defined
\begin{equation}
     \begin{split}
       &  \mu=\frac{m_1}{10}+\frac{1}{96} (2 m_2+3 m_3+\frac{4m_4}{3} +4 m_6)\,,\quad\nu=\frac{n_1}{10}+\frac{n_3}{32}+\frac{n_4}{72}+\frac{n_5+n_6}{24} \,,\\
       & \lambda=\frac{\ell_1}{10}+\frac{1}{96} \left(2 \ell_2+3 \ell_3+\frac{4\ell_3}{3}+4 (\ell_5+\ell_6)\right)\,.
    \end{split}
\end{equation}
Finally, we can formulate the superpotential as
\begin{equation}\label{eq:finalsuperpot}
W(S,M)=\frac{\Omega(S)\mc{H}(M)}{\chi_\gamma(M)}\,,
\end{equation}
where $\chi_\gamma\equiv\chi_\beta^{\alpha}$ is a weight-$\gamma$ cusp form. We will use \cref{eq:finalsuperpot} in the rest of this work, choosing $\gamma$ appropriately.

\section{Scalar potential}
\label{sec:scalarpot}
We now have all the ingredients to compute the F-term scalar potential as in~\cref{eq:Fpotential} and thus investigate its properties. To be concrete, we will consider a model with the four moduli $T,U,Z$ and $S$ with K\"{a}hler potential
\begin{equation}
    K = -\gamma\ln(Y) + k(S,\bar{S})\,,
\end{equation}
and the superpotential in~\cref{eq:finalsuperpot}. To construct the scalar potential, we first examine the F-terms 
$F_I=\partial_I W+W\partial_I K$ for $I=S,\,T,\,U,\,Z$,  where from now on we denote for fields $\Phi_I$ derivatives $\partial/\partial\Phi_I$ as $\partial_I$. The F-term for the dilaton reads
\begin{equation}
    \begin{aligned}
    F_S&=\frac{\sew{H}}{\chi_\gamma}\left(\partial_S\Omega+(\partial_Sk)\Omega\right)=W\left(\partial_S\ln\Omega+\partial_Sk\right)\equiv \frac{\mathcal{H}}{\chi_\gamma}\widetilde{F}_S\,,
    \end{aligned}
\label{eq:dil_F}
\end{equation}
while those for the geometric fields $M_{ij}$ read
\begin{equation}
    \begin{aligned}
    F_{M_{ij}}
    &=W\bigg( \partial_{M_{ij}}\log(\mathcal{H})- \frac{\nabla_{\!M_{ij}}\chi_\gamma}{\chi_\gamma}\bigg)= \frac{\Omega(S)}{\chi_\gamma} \widetilde{F}_{M_{ij}}\,,
    \end{aligned}
\label{eq:geo_F}
\end{equation}
where we have used the definition of the Siegel covariant derivative from the Appendix. From the discussion there, we see that the re-scaled F-terms $\widetilde{F}_{M_{ij}}$ defined in the last equality of~\cref{eq:geo_F} transform as bipartite Siegel modular forms. These same re-scaled F-terms will be used in the no-go theorems in the subsequent section. 

Since the K\"{a}hler metric is block-diagonal between the geometric and dilaton moduli, we can write the scalar potential~\cref{eq:Fpotential} as
\begin{equation}
    \begin{aligned}
        V &= \frac{e^k}{Y^\gamma}\bigg(F_S\bar{F}_{\bar{S}}K^{S\bar{S}} + K^{a\bar{b}}F_a\bar{F}_{\bar{b}} - 3W\bar{W}\bigg)\\
        &= \frac{e^k}{Y^\gamma|\chi_\gamma|^2} \bigg( \big| \mathcal{H}\big|^2\left(K^{S\bar{S}}\widetilde{F}_S\widetilde{\bar{F}}_{\bar{S}}-3\right)  + \big| \Omega(S)\big|^2 K^{a\bar{b}}\widetilde{F}_a\widetilde{\bar{F}_{\bar{b}}}\bigg)\\
        &= e^k \mc{Z}(M) \lvert\Omega(S)\rvert^2\bigg((A(S,\bar{S})-3)\lvert\mathcal{H}(M)\rvert^2 +\hat{V}(M,\bar{M})\bigg)\,,
    \end{aligned}
\label{eq:4mod_pot}
\end{equation}
where the subscripts take values $a,b = T,U,Z$ and we have defined~\cite{Cvetic:1991qm,Leedom:2022zdm} 
\begin{equation}\label{eq:potblocks}
    \begin{split}
        \mc{Z}(M) &:= \frac{1}{Y^\gamma \rvert\chi_\gamma\lvert^2}\,,\\
        A(S,\bar{S}) &:= K^{S\bar{S}}\big| \partial_S\log(\Omega(S)) +k_s\big|^2\,,\\
        \hat{V}(M,\bar{M}) &:= K^{a\bar{b}}\widetilde{F}_a\widetilde{\bar{F}}_{\bar{b}}\,.
    \end{split}
\end{equation}
As an explicit check, we can easily see how the above structure reduces back to the familiar $\slz$ result by first taking the $U\to\infty$ limit and only  then switching off the Wilson line.  Note that the order of these limits is important. Indeed, at finite Wilson line, the degeneration at $U\to \infty$ is not directly to an ordinary elliptic modular form, but rather to the Fourier-Jacobi expansion of a weight-$k$ genus-2 form $F_k$:
\begin{equation}
    F_k(M)=\sum_{m\geq 0}\phi_{k,m}(T,Z) e^{2\pi i m U} \,,
\end{equation}
where each coefficient $\phi_{k,m}(T,Z)$ is a Jacobi  form of weight $k$ and index $m$.
In the intermediate regime the physics is described by Jacobi forms, whose variable keeps track of the Wilson line dependence. Only after the Jacobi form has been isolated, the further limit of $Z\rightarrow0$ projects the system onto the purely classical modular sector, thereby recovering the standard $\slz$ result. 

\section{Zeroes of Siegel modular forms and no-go theorems}
\label{sec:nogos}
To study the extrema of the scalar potential, we prove some novel results on the zeros of  SMFs of degree 2. These will be useful to understand the behaviors of SMFs at the elliptic points.

The elliptic fixed points of $\spz$ were categorized in~\cite{Gottschling1961a,Gottschling1961b,Gottschling1967} and recently summarized in~\cite{Ding:2024xhz,Bashmakov:2022uek}. We define
\begin{equation}
    \zeta = e^{2\pi i/5}\,,\quad
    \eta = \frac{1}{3}(1+i\,2\sqrt{2})\,,\quad
    \rho = e^{2\pi i/3}\,.
\end{equation}
The elliptic fixed points $\sigma_i$ and their stabilizer subgroups Stab$(\sigma_i)\subset\spz$ of order $\ord(\sigma_i)$  with generators $h_a^{(i)}$ are listed in~\cref{tab:fixedpts}.
\begin{table}[t!]
\centering
\begin{tabular}{|c|c|c|c|c|c|c|c|}
\hline
$\sigma_i$ & Stab$(\sigma_i)$ & $\! $$\!$$\ord(\sigma_i)$$\!$$\!$$\!$& $\!$$\!$Generators$\!$$\!$&$\!$$\sef$$\!$&$\!$$\ses$$\!$&$\!$$\!$$\chi_{10}$$\!$$\!$&$\!$$\!$$\chi_{12}$$\!$$\!$\\
\hline\hline
$\!$$\!$$ \displaystyle
     \sigma_1=\begin{pmatrix}
                    \zeta  & \zeta + \zeta^{-2}\\   \zeta + \zeta^{-2} & - \zeta^{-1}
                \end{pmatrix}$$\!$$\!$ &
 $\mathbb{Z}_{10}$ & 10
   & $h^{(1)}_1$
   & 0 & 0 & -- & 0
   \\[1.5em]
  $\!$$\displaystyle
    \sigma_2= \begin{pmatrix}
                    \eta  & \frac{1}{2}(\eta \!-\!1)\\
                    \frac{1}{2}(\eta \!-\!1) & \eta
                \end{pmatrix}$$\!$$\!$$\!$ & $GL(2,3)$ & 48 & \shortstack[c]{$h^{(2)}_1,\, h^{(2)}_2$} 
                & -- & -- & -- & --
   \\[1.5em]
 $\displaystyle\sigma_3=\operatorname{diag}(i,i)$
                & $(\mathbb{Z}_{4}\times \mathbb{Z}_4)\rtimes \mathbb{Z}_2$ & 32
   &$\!$$\!$$\!$
  \shortstack[c]{
      $\!$$h^{(3)}_1,h^{(3)}_2,h^{(3)}_3$}$\!$$\!$
    & -- & 0 & 0 & --
   \\[1.5em]
$\displaystyle\sigma_4= \operatorname{diag}(\rho,\rho)$
               & $\!$$\!$$\mathbb{Z}_3\times(\mathbb{Z}_{6}\times \mathbb{Z}_2)\rtimes \mathbb{Z}_2$$\!$$\!$ &  72
   &$\!$$\!$ $\!$\shortstack[c]{$\!$$h^{(4)}_1,h^{(4)}_2,h^{(4)}_3$}$\!$$\!$
   & 0 & -- & 0 & --
   \\[1.5em]
$\displaystyle\sigma_5= \frac{i}{\sqrt{3}} \begin{pmatrix}
                    2  & 1\\
                    1 & 2
                \end{pmatrix} $
   & $(\mathbb{Z}_2\times\mathbb{Z}_6)\rtimes\mathbb{Z}_2$ & 24
   &$\!$$\!$\shortstack[c]{$\!$$h^{(5)}_1,h^{(5)}_2,h^{(5)}_3$}$\!$$\!$
   & -- & -- & -- & --
   \\[1.5em]
 $\displaystyle\sigma_6=\operatorname{diag}(\rho,i)$
   & $\mathbb{Z}_{12}\times \mathbb{Z}_2$ & 24
   & $h^{(6)}_1,\,h^{(6)}_2$
   & 0 & 0 & 0 & -- \\
\hline
\end{tabular}
\caption{Fixed points $\sigma_i$ with their stabilizer subgroups Stab$(\sigma_i)\subset\spz$ of order $\ord(\sigma_i)$, and values of the SMFs at those points (``$-$" stands for non-vanishing). The expressions for the generators can be found in~\cref{app:generators}.}
\label{tab:fixedpts}
\end{table}
In order to write the last four columns of~\cref{tab:fixedpts}, we proved the following
\begin{theorem} Let $F_k$ be a classical Siegel modular form of degree 2 and weight $k$. Furthermore, we let the multiplier system of $F_k$ be trivial for all elements of $\spz$. Let
\begin{equation}
h_a^{(i)}=\begin{pmatrix}
    A_a^{(i)}& B_a^{(i)}\\
    C_a^{(i)}& D_a^{(i)}
\label{eq:stabmatrix}
\end{pmatrix}
\end{equation}
be the $a$-th generators of Stab$(\sigma_i)\subset\spz$. Then at $\sigma_i$ we have that either det$(C_a^{(i)}\sigma_i+D_a^{(i)})^{k}=1$ or $F_k(\sigma_i)=0$.
In particular, we have
\begin{align}
    F_k (\sigma_1) &= 0 \qquad \text{unless} \qquad k= 0 \text{ mod } 5\nonumber\\
    F_k (\sigma_2) &= 0 \qquad \text{unless} \qquad k= 0 \text{ mod } 2\nonumber\\
    F_k (\sigma_3) &= 0 \qquad \text{unless} \qquad k= 0 \text{ mod } 4\nonumber\\
    F_k (\sigma_4) &= 0 \qquad \text{unless} \qquad k= 0 \text{ mod } 6\nonumber\\
    F_k (\sigma_5) &= 0 \qquad \text{unless} \qquad k= 0 \text{ mod } 2\nonumber\\
    F_k (\sigma_6) &= 0 \qquad \text{unless} \qquad k= 0 \text{ mod } 12\nonumber
\end{align}
\end{theorem}
\begin{proof}
    We first define
        \begin{equation}
            \text{det}^{(i)}_a = \text{det}[C_a^{(i)}\sigma_i + D_a^{(i)}]\coma
        \end{equation}
     where $C^{(i)}_a$ are constituent $2\times 2$ matrices in $h^{(i)}_a$, the $a$-th generator of Stab$_{(i)}$.
    Then the classical Siegel modular form satisfies
        \begin{equation}
            F_k(h^{(i)}_a \cdot \sigma_i) = (\text{det}^{(i)}_a)^k F_k(\sigma_i) = F_k(\sigma_i)\,,
        \end{equation}
    where the first equality follows from the transformation law of Siegel modular forms and the second from the definition of a fixed point and its stabilizer group. Thus we see that necessarily $F_k(\sigma_i) =0$ unless $(\text{det}^{(i)}_a)^k =1$ $\forall a$. For the elliptic fixed points (described in~\cref{tab:fixedpts}) we find
 \begin{equation}
     \begin{aligned}
         \sigma_1\,:& \quad  \text{det}_1 = e^{2\pi i /5}\\
         \sigma_2\,:& \quad  \text{det}_1 = -1\coma
        \text{det}_2 = -1\\
         \sigma_3\,:& \quad  \text{det}_1 = i\coma
        \text{det}_2 = 1 \coma
        \text{det}_3 = -1\\
        \sigma_4\,:& \quad  \text{det}_1 = -1\coma
        \text{det}_2 = -1 \coma
        \text{det}_3 = e^{-i\pi/3}\\
         \sigma_5\,:& \quad  \text{det}_1 = 1\coma
        \text{det}_2 = 1 \coma
        \text{det}_3 = -1\\
        \sigma_6\,:& \quad  \text{det}_1 = e^{5\pi i /6}\coma \text{det}_2 = 1\,.
     \end{aligned}
 \end{equation}
 With these values, the conditions on $k$ for non-vanishing values at the fixed points follows directly from the  requirement $(\text{det}^{(i)}_a)^k =1$ $\forall a$. 
\end{proof}
\noindent As is evident from the assumptions of this theorem, the above results do not hold for Siegel modular forms with non-trivial multiplier systems. However, it is clear that identical steps can be carried out for such forms. 

As it stands, the above theorem allows us to determine which forms of the ring vanish at which elliptic fixed points. We can also prove the following
\begin{theorem}\label{thm:vanishbip}
    Let $\mathbf{F}(\Omega)$ be a bipartite Siegel modular form of genus 2 and trivial multiplier system as defined in~\cref{app:SMFs}, with transformation property
        \begin{equation}
            \mathbf{F}(\Gamma\cdot M) = (C M +D) \mathbf{F}(M) (C M+D)^T
        \label{eq:prooftransform}
        \end{equation}
    for all $\Gamma\in \spz$. Then $\mathbf{F}(M)$ vanishes at the fixed points $\sigma_i$, $i=1,..,6$ of $\spz$.
\end{theorem}
\begin{proof}
    We first vectorize the transformation property so that~\cref{eq:prooftransform} becomes
    \begin{equation}
        \text{vec}[\mathbf{F}(\Gamma\cdot M)] 
            = \Big((C M +D) \otimes (C M +D)\Big)\text{vec}[\mathbf{F}(M)]\,.
    \label{eq:proofvec}
    \end{equation}
    As described above, the fixed points $\{\sigma_i\}$ each have stabilizer groups $\text{Stab}(\sigma_i)$ with generators $h_a^{(i)}$, $a = 1,..,\text{ord}(\text{Stab}(\sigma_i))$, such that $h_a^{(i)}\cdot\sigma_i = \sigma_i$. Thus for each of the generators of $\text{Stab}(\sigma_i)$,~\cref{eq:proofvec} becomes
        \begin{equation}
        \text{vec}[\mathbf{F}(\sigma_i)] 
            = \Big((C_a^{(i)}\sigma_i +D_a^{(i)}) \otimes (C_a^{(i)}\sigma_i +D_a^{(i)})\Big)\text{vec}[\mathbf{F}(\sigma_i)]\,.
    \label{eq:proofeigen}
    \end{equation}
\cref{eq:proofeigen} defines an eigenvalue problem. Non-trivial solutions for $\text{vec}[\mathbf{F}(
    \sigma_i)]$ exist if and only if the determinant of 
    \begin{equation}
        \begin{aligned}
            M^{(i)}_a := \bigg[1_4 -     \Big((C_a^{(i)}\sigma_i +D_a^{(i)}) \otimes (C_a^{(i)}\sigma_i +D_a^{(i)})\Big)\bigg] 
        \end{aligned}
    \label{eq:proofdet}
    \end{equation}
    vanishes. This requires the vanishing of one or more of the eigenvalues of the matrix $M^{(i)}_a$. The eigenvalues of $M^{(i)}_a$ have the form $1 - \lambda_i\lambda_j$, where the $\lambda_i$ are eigenvalues of $(C_a^{(i)}\sigma_i +D_a^{(i)})$. Therefore, a non-trivial solution to~\cref{eq:proofeigen} exists only if the product of two eigenvalues of $(C_a^{(i)}\sigma_i +D_a^{(i)})$ is equal to unity. 
    
    Before continuing the analysis, we stop to emphasize a crucial point. The vanishing of $\text{det}(M^{(i)}_a)$ is sufficient to determine non-trivial solutions to~\cref{eq:proofeigen} \textit{for a given $M^{(i)}_a$}. Since several fixed points $\{\sigma_i\}$ have more than one generator in their stabilizer group, non-trivial values of the Siegel modular form exist only if there exists a vector $\text{vec}[\mathbf{F}(\sigma_i)]$ that satisfies~\cref{eq:proofeigen} for each value of $a$.
    
    We now turn to the eigenvalues of the $(C_a^{(i)}\sigma_i +D_a^{(i)})$ matrices, which are displayed in~\cref{tab:generator_eigs}, found in~\cref{app:SMFs}. By inspection, we see that indeed there are cases where the eigenvalues of $M^{(i)}_a$ vanish. We can create three categories of fixed points:
    \begin{itemize}
        \item Category \#1: For a fixed $\sigma_i$, none of the matrices $M^{(i)}_a$ have a zero eigenvalue
        \item Category \#2: For a fixed $\sigma_i$, there exists one value of $a\in 1,..,\text{ord}(\text{Stab}(\sigma_i)$ such that the corresponding $M^{(i)}_a$ does not have any zero eigenvalues.
        \item Category \#3: For a fixed $\sigma_i$, every matrix $M^{(i)}_a$ has at least one zero eigenvalue.
    \end{itemize}
    Inspecting the values in~\cref{tab:generator_eigs}, we see that fixed points $\{\sigma_1,\sigma_6\}$ are in Category \#1, points $\{\sigma_2,\sigma_4\}$ are in Category \#2, and finally points $\{\sigma_3,\sigma_5\}$ are in Category \# 3.

    It then follows immediately that at the fixed points $\sigma_1,\sigma_2,\sigma_4,\sigma_6$, the function $\mathbf{F}(M)$ vanishes. We must then consider the remaining points of Category \#3. By definition, the points in this category have matrices $M^{(i)}_a$ with vanishing eigenvalues for every value of $a$. However, $\mathbf{F}(M)$ can only attain non-null values if, for fixed $i$ and all $a$, the matrices $M^{(i)}_a$ share at least one eigenvector with null eigenvalue. The relevant eigenvectors as displayed in~\cref{tab:generator_eigenvects}. We see that the sole null eigenvector of $M^{(3)}_1$ is not shared by $M^{(3)}_2$ or $M^{(3)}_3$. Thus there is no common eigenvector and no non-trivial solutions $\mathbf{F}(\sigma_3)$ can exist. A straightforward analysis reveals that no linear superposition of any of the pairs of null eigenvectors of the $M^{(6)}_a$ is a common eigenvector for all $a$, and so $\mathbf{F}(\sigma_3)$ cannot have any non-trivial solutions either.   
\end{proof}
\noindent We have an immediate corollary:
\begin{corollary}\label{cor:extrema}
Let $F(M)$ be a Siegel modular function with trivial multiplier system. Then $F(M)$ will have extrema at the fixed points $\sigma_i$, $i=1,\dots,6$ of $\spz$.
\end{corollary}
\begin{proof}
The result follows immediately from~\cref{thm:vanishbip} by applying the matrix-derivative $\partial_M$ to $F(M)$. The resulting matrix of first derivatives transforms under $\spz$ as 
\begin{equation}
    \partial_MF(\Gamma\cdot M) = (CM+D)\partial_MF(M)(CM+D)^T\,.
\end{equation}
Therefore, $\partial_M F$ is a bipartite vector-valued Siegel modular function and the results of~\cref{thm:vanishbip} apply. 
\end{proof}
\noindent We can now use Corollary~\ref{cor:extrema} to deduce that the quantities
\begin{equation}
    \partial_M V\,,\; \partial_{M_{ij}} G\,,\; \partial_{M_{ij}}\partial_S V\,,\; \widetilde{F}_{i}
\end{equation}
must all vanish at the fixed points of $\spz$.

Equipped with these new results, we can prove generalizations of the no-go theorems in~\cite{Leedom:2022zdm}:
\begin{theorem}
    At a point $(S_0,M_0)$, where $M_0$ is one of the fixed points of $\spz$, the scalar potential~\cref{eq:4mod_pot} cannot have de sitter vacua if   
        \begin{equation}
            \widetilde{F}_S(S_0) = 0\,.
        \end{equation}
\end{theorem}
\begin{proof}
    This statement is in essence an extension of Theorem 1 in~\cite{Leedom:2022zdm}. We will not be as general as that theorem and will instead focus purely on the statement above. The statement follows from the fact that if $\widetilde{F}_S(S_0) =0$ then~\cref{eq:4mod_pot} becomes
    \begin{equation}
        V_0 = e^{k_0} \mc{Z}_0\lvert\Omega_0\rvert^2\bigg( \hat{V}_0 -3\lvert\mathcal{H}_0\rvert^2\bigg)\,,
    \end{equation}
    where the subscript $0$ indicates that the quantity has been evaluated at $(S_0,M_0)$. However, every term in $\hat{V}$ contains the re-scaled F-terms of the geometric moduli $T,U,Z$. These form a bipartite Siegel modular form, and therefore they must vanish at the fixed points of $\spz$. Hence when $\widetilde{F}_S=0$, the scalar potential at the fixed points becomes
    \begin{equation}
        V_0 = -3 e^{k_0}\mc{Z}_0\lvert \Omega_0\rvert^2 <0\,.
    \end{equation}
\end{proof}
\noindent Next, we have
\begin{theorem}
    At a point $(S_0,M_0)$, where $M_0$ is one of the fixed points of $\spz$, the scalar potential~\cref{eq:4mod_pot} with the dilaton K\"{a}hler potential
    \begin{equation}
        k(S,\bar{S}) = -\ln(S+\bar{S})
    \end{equation}
    can not simultaneously satisfy
    \begin{itemize}
        \item[(i)] $V(S_0,M_0) >0$,
        \item[(ii)] $\partial_SV(S_0,M_0) =0$,
        \item[(iii)] Eigenvalues of the Hessian of the potential at $(S_0,M_0)$ are all $\ge0$.
    \end{itemize}
\end{theorem}
\begin{proof}
    This is an extension of Theorem 2 of~\cite{Leedom:2022zdm} to $\spz$, although we have changed the wording slightly to focus on the fixed points. The proof proceeds by contradiction. For the moment, we will keep the dilaton K\"{a}hler potential as a general function and then plug in the assumption $k(S,\bar{S}) = - \ln(S+\bar{S})$ when required. Since $M_0$ is a fixed point, we know from Corollary~\ref{cor:extrema} that
    \begin{equation}
        \begin{aligned}
           & \partial_{M_{ij}} V(S_0,M_0) = 0\,,\\
           & \widetilde{F}_T(M_0) = \widetilde{F}_U(M_0) =\widetilde{F}_Z(M_0) = 0\,,\\
            &\hat{V}_0 = 0\,.
        \end{aligned}
    \end{equation}
Using assumption (i), we can introduce a positive number $\Lambda^4>0$ such that
\begin{equation}
    V(S_0,M_0) = e^{k_0}\mc{Z}_0 \lvert\mathcal{H}_0\rvert^2\Lambda^4\,.
\label{eq:thrm2_pot}
\end{equation}
Note that $\mathcal{H}_0$ must be non-zero using assumption (i) since $\hat{V}_0 = 0$.~\cref{eq:thrm2_pot} can be solved to yield an expression for the first derivative of $\Omega(S)$:
\begin{equation}
    (\Omega_S)_0 = -(k_S)_0\Omega_0 \pm \sqrt{(k_{S\bar{S}})_0}\,\bigg\{\Lambda^2 \pm i\bigg(\frac{3\lvert\Omega_0\rvert^2\lvert\mathcal{H}_0\rvert^2 - \lvert\Omega_0\rvert^2\hat{V}_0}{\lvert\mathcal{H}_0\rvert^2}\bigg)^{\frac12}\bigg\}\,.
\label{eq:om_1div}
\end{equation}
Similarly, assumption (ii) can be used to determine an expression for the second derivative of $\Omega$ at $S_0$:
\begin{equation}
\begin{split}
    (\Omega_{SS})_0 &=  \bar{\Omega}_0 e^{2i\sigma_0}(k_{S\bar{S}})_0 \bigg(2 - \frac{1}{\lvert\mathcal{H}_0\rvert^2}\hat{V}_0\bigg) + \bigg(\frac{k_{SS\bar{S}}}{k_{S\bar{S}}}-k_S\bigg)_{\!0}(\widetilde{F}_S)_0 \\& \;\;\;\;- (k_{SS})_0\Omega_0 - (k_S)_0(\Omega_S)_0\,,
\end{split}
\label{eq:om_2div}
\end{equation}
where $\sigma_0 = \text{arg}((\Omega_S)_0+(k_S)_0\Omega_0)$. We will use~\cref{eq:om_1div} and~\cref{eq:om_2div} below.  Next, we note that
\begin{equation}
    \partial_{M_{ij}}V = e^k\mc{Z}\lvert\Omega\rvert^2 \bigg\{\bar{\mathcal{H}}(A(S,\bar{S})-3)\widetilde{F}_{M_{ij}} +\partial_{M_{ij}}\hat{V} - \bigg(\frac{\partial_{M_{ij}}\chi_\gamma}{\chi_\gamma} +\gamma\,\partial_{M_{ij}}\ln(Y)\bigg)\hat{V}\bigg\}\,.
\label{eq:partial_pot}
\end{equation}
 The derivatives in~\cref{eq:partial_pot} vanish at $(S_0,M_0)$, as argued above from Corollary~\ref{cor:extrema}. However we note that each term in~\cref{eq:partial_pot} is proportional to one of the geometric moduli re-scaled F-terms, which also vanish at the fixed points. Therefore, we have
\begin{equation}
    \partial_S^k \partial_{\bar{S}}^l\partial_{M_{ij}}V(S_0,M_0) = 0
\end{equation}
for all positive integers $k,l$. This implies that the Hessian matrix is block diagonal, with one block for the geometric moduli $T,U,Z$ and another block for the dilaton. The components for the dilaton block of the Hessian are
\begin{equation}
    \begin{aligned}
        \partial_s^2V &= 2\partial_S\partial_{\bar{S}}V + 2\text{Re}(\partial_S^2V)\,,\\
        \partial_b^2V &= 2\partial_S\partial_{\bar{S}}V - 2\text{Re}(\partial_S^2V)\,,\\
        \partial_s\partial_b V &= -2\text{Im}(\partial_S^2V)\,,
    \end{aligned}
\end{equation}
where we have split $S=s+ib$. Using~\cref{eq:om_1div} and~\cref{eq:om_2div}, as well as $\hat{V}_0=0$ and the assumption $k(S,\bar{S}) = -\ln(S+\bar{S})$, we find
\begin{equation}
    \partial_S\partial_{\bar{S}}V(S_0,M_0) = -\frac{2\mc{Z}_0\lvert\mathcal{H}_0\rvert^2(2\lvert\Omega_0\rvert^2+\Lambda^4)}{(S_0+\bar{S}_0)^3}\,.
\end{equation}
Since we require that $\text{Re}(S_0)>0$ for a sensible coupling constant, we have that $\partial_S\partial_{\bar{S}} V(S_0,M_0) < 0$. The trace of the dilaton block of the Hessian is then simply $2\partial_S\partial_{\bar{S}}V <0$, which implies that at least one eigenvalue of this block must be negative and assumption (iii) is violated.
\end{proof}

\noindent As in~\cite{Leedom:2022zdm}, we conclude that to find minima with a positive value of the potential, we shall break supersymmetry in the dilaton direction. This can be achieved by taking the function $A(S,\bar{S})\neq 0$, cf.~\cref{eq:potblocks}.

\section{Minimizing the scalar potential}\label{sec:minima}
We will now look at the full scalar potential of the $STUZ$ theory, in particular in the presence of non-vanishing dilaton F-term implying $A(S,\bar S)>0$. For this purpose, it is useful to write out the scalar potential in~\cref{eq:4mod_pot} in a slightly less compact form which shows the structure of the terms pertaining to the different moduli explicitly. 

The F-term scalar potential reads
\begin{equation}\label{eq:scalarpot}
\begin{split}
    V=&\,\frac{e^{k(S,\bar{S})}|\Omega(S)|^2}{Y^3|\chi_\gamma|^{2}}\Bigg\{\left(A(S,\bar{S})-3\right)|\sew{H}|^2+\hat{V}^{(T)}+\hat{V}^{(U)}+\hat{V}^{(Z)}\\
    &+\frac{8}{\gamma}(\Im Z)^2\Re\left[\left(\partial_T\sew{H}-\sew{H}\,\frac{\nabla_{\!T}\chi_\gamma}{\chi_\gamma}\right)\left(\partial_{\bar{U}}\bar{\sew{H}}-\bar{\sew{H}}\,\frac{\nabla_{\!\bar U}\bar\chi_\gamma}{\bar\chi_\gamma}\right)\right]\\
    &+\frac{8}{\gamma} \Im T \Im Z 
    \Re\left[\left(\partial_Z\sew{H}-\sew{H}\,\frac{\nabla_{\!Z}\chi_\gamma}{\chi_\gamma}\right)\left(\partial_{\bar{T}}\bar{\sew{H}}-\bar{\sew{H}}\,\frac{\nabla_{\!\bar T}\bar\chi_\gamma}{\bar\chi_\gamma}\right)\right]\\
    &+\frac{8}{\gamma}\Im U \Im Z \Re\left[\left(\sew{H}_Z-\sew{H}\,\frac{\nabla_{\!Z}\chi_\gamma}{\chi_\gamma}\right)\left(\partial_{\bar{U}}\bar{\sew{H}}-\bar{\sew{H}}\,\frac{\nabla_{\!\bar U}\bar\chi_\gamma}{\bar\chi_\gamma}\right)\right]
    \Bigg\}\,,
\end{split}
\end{equation}
where
\begin{equation}
\begin{split}
    &\hat{V}^{(T)}=-\frac{(T-\bar T)^2}{\gamma}\left|\partial_T\sew{H}-\sew{H}\,\frac{\nabla_{\!T}\chi_\gamma}{\chi_\gamma}\right|^2\,,\\& \hat{V}^{(U)}=-\frac{(U-\bar U)^2}{\gamma}\left|\partial_U\sew{H}-\sew{H}\,\frac{\nabla_{\!U}\chi_\gamma}{\chi_\gamma}\right|^2\,,\\
    & \hat{V}^{(Z)}=\frac{2}{\gamma}\left(\Im T \Im U+\Im Z^2\right)\left|\partial_Z\sew{H}-\sew{H}\,\frac{\nabla_{\!Z}\chi_\gamma}{\chi_\gamma}\right|^2\,.
\end{split}
\end{equation}
Note that $\hat{V}^{(M_{ij})}\geq 0$, and that the overall structure of $V$ is remarkably similar to that of the $\slz$-invariant potential derived in~\cite{Cvetic:1991qm}. Moreover, as shown in~\cite{Leedom:2022zdm} for $\slz$-symmetric setups, for $A(S,\bar{S})=0$, the scalar potential is always negative. We will drop this assumption below.

To compute the minima of this scalar potential with the ingredients described in the present work, we focus again on the six elliptic fixed points. To perform the numerical evaluations of this section (and in general of the present work) we have coded the Mathematica package~\hyperlink{https://github.com/nrighi/smfs.git}{SMFs.m}~\cite{righi:2026SMF}. This package is based on the discussion in Appendix C of~\cite{Klemm:2015iya}, which nicely summarizes the results of~\cite{eichlerzagier} (for an independent code realizing the same Appendix C of~\cite{Klemm:2015iya} in Pari/GP see~\cite{kidambi:2023SMF}).

First, we compute the Hessian matrix and its eigenvalues at all the six fixed points for $A(S,\bar{S})=0$. In this case, we find that the signs of the eigenvalues are
\begin{equation}\label{eq:eigenval}
    \begin{split}
        &\sigma_1\,: (-,-,-)\,,\qquad\sigma_2\,: (+,+,+)\,,\\
        &\sigma_3\,: (+,+,-)\,,\qquad\sigma_4\,: (-,-,-)\,,\\
        &\sigma_5\,: (-,-,-)\,,\qquad\sigma_6\,: (-,-,+)\,.
    \end{split}
\end{equation}
We do not show the exact numerical values as they depend on the free parameters of the theory (gauge groups etc.); different values of the free parameters will not change the overall sign though. We also remark that for $A(S,\bar{S})=0$ and no racetrack in the superpotential, the dilaton is a flat direction. 

Note that even if some eigenvalues are negative, they do not signal a tachyon in the theory. For $A(S,\bar{S})=0$ the minimum of the potential always corresponds to a negative cosmological constant, and the absolute value of those eigenvalues is always above the Breitenlohner-Freedman bound.

\begin{figure}
    \centering
    \includegraphics[width=0.97\linewidth]{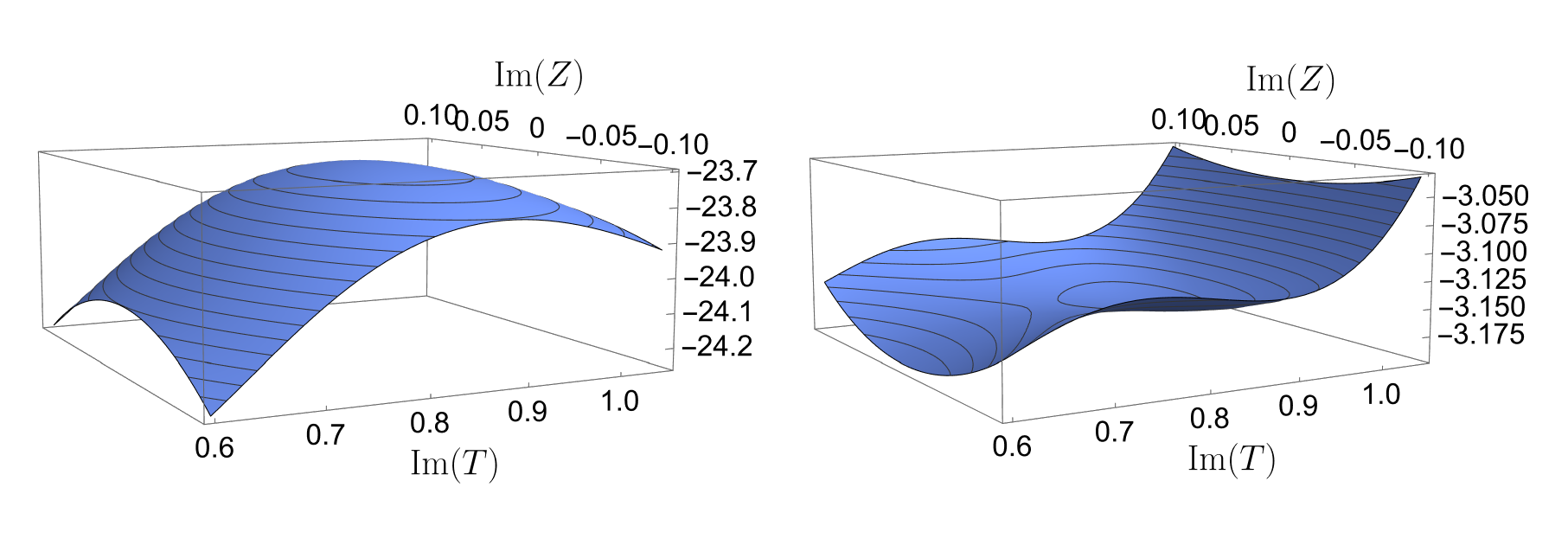}\vspace{-20pt}
    \includegraphics[width=0.47\linewidth]{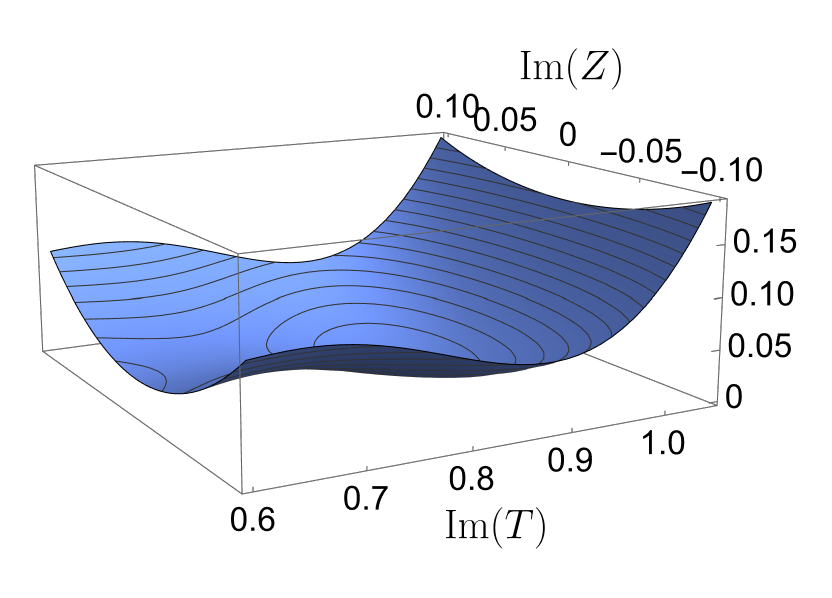}
    \caption{The scalar potential for three different values of the function $A(S,\bar S)$ around the elliptic fixed point $\sigma_4$ (for which $T=U$). We  set $\mc{H}(M)=1$ for ease of computation. First line left: $A(S,\bar S)=0$ (maximum in AdS). First line right:  $A(S,\bar S)=2.6$ (minimum in AdS). Second line: $A(S,\bar S)=3.001$ (minimum in dS).}
    \label{fig:potential4}
\end{figure}

\begin{figure}[ht]
    \centering  \includegraphics[width=0.7\linewidth]{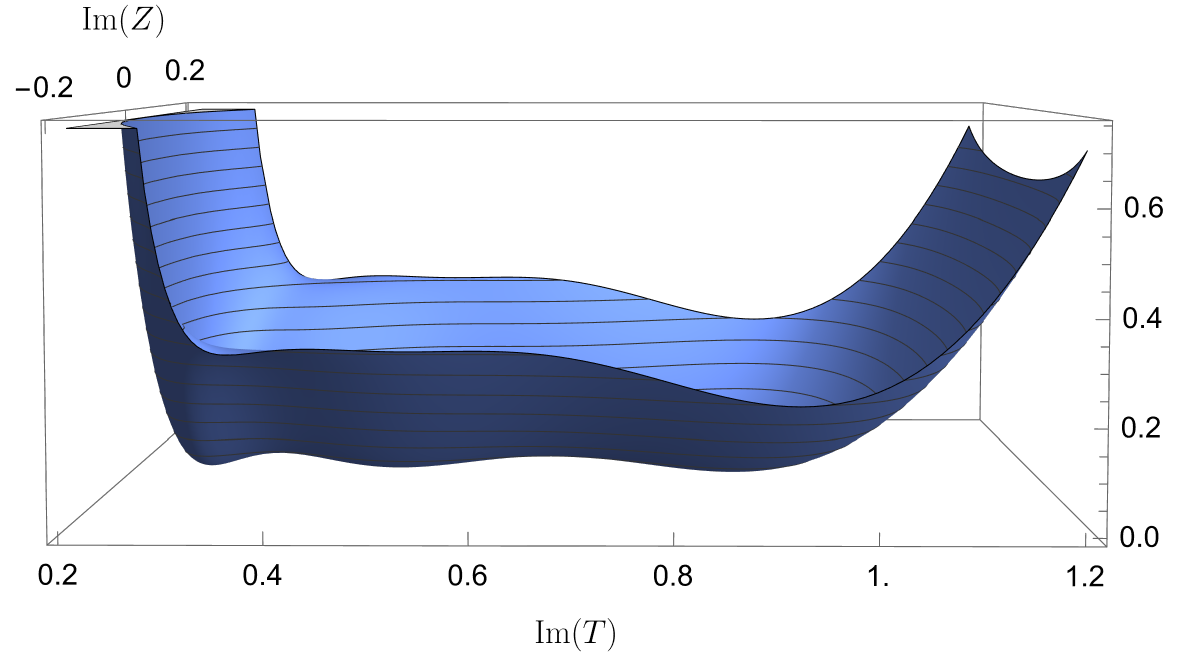}
    \caption{Potential away from elliptic fixed point $\sigma_4$, for $\mc{H}(M)=\Omega(S)=1$, $A(S,\bar{S})=3.001$ and $U=U|_{\sigma_4}$. The three minima lie at $T=T|_{\sigma_4}$ and $\Im T=0.287,\, 0.5$, both with $\Re T=-0.5$ and $Z=0$.}
    \label{fig:potoutrim}
\end{figure}

Nevertheless, a negative mass-squared can be concerning when attempting an uplift to de Sitter, unless this mass changes sign during the uplift. The latter is precisely what happens when we include nonperturbative corrections to the dilaton in the K\"ahler potential (the $k_{np}(S)$ in eq.~\cref{eq:kpot}). When we include a $A(S,\bar{S})>2$, all the negative eigenvalues in \cref{eq:eigenval} become positive, and for $A(S,\bar{S})>3$ we have (metastable) minima with a positive cosmological constant. We show this process in \cref{fig:potential4}. This behavior displays how, by changing the potential, you change not only the geometry but also the spectrum. 

Finally, a comment about the behavior at large $T$ (or $U$) is in order. In such regions, the potential diverges because the cusp form at the denominator in the superpotential is parametrizing the whole tower of KK states. Indeed, at large distances from the interior of the moduli space, those states become light, rendering the original perturbative description of the theory invalid. This behavior  is very well known in theories with $\slz$ modular symmetry~\cite{Cvetic:1991qm,Gonzalo:2018guu,Leedom:2022zdm}, and it is remarkable that threshold corrections to the gauge coupling capture a whole tower of states; as a consequence, the breakdown of the theory is manifestly evident since we can clearly tell when the EFT ceases to be valid. Note that this divergence is only present in the geometric moduli: the dilaton direction is a decompactification limit as well (towards heterotic M-theory in the case of the E$_8\times$ E$_8$ gauge group), and clearly it is the leading decompactification limit here. If we were to plot the same potential including the dilaton direction, we would see the runaway behavior at large distances typical of all 4D supergravities.\footnote{Relatedly, \cite{Cvetic:2024wsj} argued that a singular boundary diagnosed through threshold corrections should not be trusted within the original perturbative effective theory, since nonperturbative effects become essential and resolve the would-be singular behavior.}

Let us now try to venture inside the fundamental domain. To do so, we rely on a numerical search for extrema. We did not perform a comprehensive scan; in what follows, we present some interesting features we have found.

First, we focus again on the domain around the elliptic point $\sigma_4$ shown in \cref{fig:potential4}. By moving away from $\sigma_4$ in the $\Im Z-\Im T$ directions and keeping $A(S,\bar{S})=0$, we encounter two maxima at $T=T|_{\sigma_4}$ (cf.~\cref{fig:potential4}) and $\Im T=0.287$, connected by a saddle point at $\Im T=0.5$. By turning on $2<A(S,\bar{S})<3$, those extrema all become minima in an AdS background, while for $A(S,\bar{S})>3$ they all turn into minima in dS. The theory is $\mathbb{Z}_2$-symmetric around $Z=0$, and the presence of divergences for both limits of $\im{T}$ nicely reproduces the behavior of the $\slz$ counterpart \cite{Cvetic:1991qm,Leedom:2022zdm,Cribiori:2023sch}.

\section{Conclusions}
Modular symmetries have proved a powerful tool for studying a number of aspects in string theory. In this work, we have employed their properties  to develop a systematic framework for heterotic orbifold compactifications with a Wilson line turned on, which enhances the target-space duality group from the familiar genus 1 $\slz$ modular group to the genus 2 modular group $\spz$, and the natural language becomes that of degree 2 Siegel modular forms. The moduli are organized into a period matrix, allowing us to employ modular invariance for studying perturbative and nonperturbative contributions to the K\"ahler and superpotential, as well as the resulting scalar potential. Once the duality symmetry is treated as fundamental, many of the otherwise ambiguous choices that enter the effective description (such as which functions may appear, where singularities are allowed, which loci signal additional massless states, and how different corners of moduli space are related) become tightly correlated and can be addressed in an elegant and rigorous way.

A central feature of the EFT is the moduli dependence of one-loop threshold corrections. In genus 1, the one-loop corrected gauge kinetic function is fixed by modular symmetry together with the requirement that singularities only arise from additional light states at enhanced-symmetry loci. In the genus 2 setting the same logic persists, but the presence of the Wilson line introduces a new behavior: different gauge sectors can exhibit qualitatively different behavior when the Wilson line is taken to zero. Some sectors develop additional light charged states at vanishing Wilson line, while others do not. That physical distinction selects different automorphic objects in the threshold corrections and, through gaugino condensation, in the associated nonperturbative superpotential. In this way the Wilson-line modulus  enlarges the moduli space while also describing how the theory approaches a degeneration limit and whether that limit is accompanied by extra massless degrees of freedom. 

Beyond threshold corrections, we discussed systematically how to describe the residual dependence of the superpotential on nonperturbative corrections (instantons) in the geometric moduli subject to the constraints imposed by modular covariance of $W$.  
Any additional term in the superpotential has to be packaged into modular-invariant functions. For this, we propose a modular-invariant completion to $W$ in the geometric fields $T$, $U$, $Z$. We leave it as an important task for the future to either prove this proposal to be the direct analogue of the known most general holomorphic $\slz$-modular invariant function, or to construct it as a further generalization of our proposal.

With these ingredients, we constructed the full F-term scalar potential for the dilaton and the geometric moduli in a way that preserves modular invariance, and derived results that strongly constrain the vacua. We proved a set of theorems about the elliptic fixed points of the Siegel modular group to show that they are automatically critical points of modular-invariant functions. This property follows directly from the automorphic structure governing the theory. 

Combining these number-theoretical results with  supergravity yields a genus 2 analogue of the no-go theorems proved in~\cite{Leedom:2022zdm}: if supersymmetry is unbroken in the dilaton direction, then de Sitter vacua at the fixed points are excluded.

Motivated by this, we analyzed the stability properties of the fixed-point extrema by studying the spectrum and tracking how it changes under deformations of the dilaton sector. In the absence of supersymmetry-breaking dilaton dynamics, the fixed points generically sit at negative vacuum energy and exhibit a pattern of stable and unstable directions whose detailed structure depends on the specific fixed point but is robust in its qualitative features. A key result is that stringy nonperturbative corrections to the dilaton K\"ahler potential can  change this picture: they can lift negative modes in the geometric sector, turning instabilities into stable directions, and (as discussed in details in~\cite{Leedom:2022zdm}) even drive an uplift  to metastable dS minima, if strong enough. From an EFT perspective, this provides a mechanism in which dilaton-sector physics reshapes not only the vacuum energy but also the entire mass spectrum of the moduli.

Finally, we explored the vacuum structure away from symmetry-enhancement loci. While we did not attempt an exhaustive scan of the Siegel fundamental domain, our numerical search indicates that the potential can support additional extrema in the interior, connected by saddles in a way reminiscent of the genus-1 story~\cite{Novichkov:2022wvg,Leedom:2022zdm} but enriched by the Wilson-line direction and the higher-genus automorphy.

There are several immediate directions for future work. On the model-building side, it will be important to use the present analysis to construct explicit orbifold models with fully fledged gauge sectors and matter spectra, so that the threshold data, the admissible modular functions, and the dilaton dynamics are fixed rather than treated parametrically. On the vacuum-structure side, a systematic numerical survey of the Siegel fundamental domain would clarify how generic interior extrema are and what regions admit parametrically controlled uplifting. 

On the mathematical side, understanding the supersymmetric origin of the threshold behavior associated with the cusp form that remains non-vanishing when turning off the Wilson line (i.e.~$\chi_{12}$) remains an especially interesting problem. Borcherds showed~\cite{borcherds1995automorphic} that one could construct a genus-2 modular form of weight 12, called $\Phi_{12}$, via an exponential lift. Such modular form turns out to be the denominator function of the Fake Monster Lie algebra: $\Phi_{12}$ plays for the Fake Monster Lie algebra the same role that the Dedekind $\eta$ plays for affine Lie algebras. $\Phi_{12}$ has a well-defined vanishing locus with a zeroes and poles structure which is however completely different from that of $\chi_{12}$. Being able to write a SMF as the exponential lift of a weak Jacobi form is crucial to relate it to the BPS-saturated quantity which is the supersymmetric index. However, it seems that $\chi_{12}$ has no simple vanishing locus, and admits no Borcherds lift since it would violate the multiplicative symmetry of the lift~\cite{heim2013symmetries}. On the contrary, its additive lift seems to be well defined (and additive lifts do not force the zero locus to sit on special divisors). It would be interesting to find a supersymmetric object or quantity which can be written as the additive lift 
of a Jacobi form.\footnote{For forthcoming foundational mathematical work concerning the relation between the dimension-2 loci discussed here, called `Humbert surfaces' in the mathematical literature, and exponential vs additive lifts of SMFs, which goes beyond what is discussed here, see~\cite{Kidambi:2026ka}.} The output of such lifts would have a different degeneration behavior compared to SMFs with multiplicative representations (e.g.~a non-vanishing separating limit, or a structure representing sums over BPS states rather than products over roots). The physics interpretation of additive lifts are largely unexplored though. We hope to come back to this problem in the near future.


\acknowledgments
We are grateful to Ferruccio Feruglio, Andreas Malmendier, Sameer Murthy, Don Zagier for discussions, and especially to Abhiram Kidambi introducing us to the basics of degree 2 Siegel modular forms, including the concepts of Humbert surfaces, automorphic lifting, hyperelliptic curves and their principally polarized abelian varieties. We also thank Abhiram Kidambi for comments on the draft. JML and AW would like to thank Kavli IPMU for their hospitality
during the completion of this work. JML would also like to thank King's College London for their hospitality. NR thanks DESY for hospitality during various stages of this work, and the organizers and participants of the workshop ``Modular invariant Approach to the lepton and quark flavor problems: from bottom-up to top-down" in MITP. NR is supported by the ERC NOTIMEFORCOSMO, 101126304. Views and opinions expressed are, however, those of the author only and do not necessarily reflect those of the European Union or the European Research Council Executive Agency. Neither the European Union nor the granting authority can be held responsible for them. AW is partially supported by the Deutsche Forschungsgemeinschaft under Germany’s Excellence Strategy - EXC 2121 ``Quantum Universe'' - 390833306 and by the Deutsche Forschungsgemeinschaft through the Collaborative Research Center SFB1624  ``Higher Structures, Moduli Spaces, and Integrability''.

\appendix

\section{On Siegel modular forms}\label{app:SMFs}
We provide some basic definitions, results, and proofs related to Siegel modular forms used in the main text. We mostly follow Refs.~\cite{123ModForms,eichlerzagier,Klingen_1990}.

\subsection{The Siegel fundamental domain}
We can geometrically motivate the study of Siegel modular forms by considering a genus $g$ Riemann surface. Let \(\Sigma\) be such a surface, and choose a symplectic basis $\alpha_1, \cdots , \alpha_g, \beta_1, \cdots \beta_g \in H_1(\Sigma,\Z)$ with intersection pairing
$
\langle \alpha_i,\beta_j\rangle=\delta_{ij}
$. In this basis, the intersection form is represented by
\begin{equation}
J_g=
\begin{pmatrix}
0 & \mathbf I_g\\
-\mathbf I_g & 0
\end{pmatrix}\,.
\end{equation}
Let \(\omega_1,\ldots,\omega_g\) be a basis of holomorphic 1-forms on \(\Sigma\), normalized so that
\begin{equation}
\int_{\alpha_j}\omega_i=\delta_{ij}\,.
\end{equation}
The corresponding period matrix is then defined by
\begin{equation}
M_{ij}=\int_{\beta_j}\omega_i\,.
\end{equation}
A classical result shows that $M$ is symmetric and has positive definite imaginary part.  Therefore, this generalizes the notion of the classical upper half-space of genus 1 to the genus-$g$  \textit{Siegel upper half-space} $\mathbb{H}_g$:
\begin{equation}
    \mathbb{H}_g := \left \{M \in \text{Mat}(g\times g,\C) \middle\vert M^T = M, \ \Im(M) >0 \right \}\,.
\end{equation}
The symplectic basis can be changed by the action of the symplectic group 
\begin{equation}
    \spg := \left \{ M \in \text{Mat}(g \times g,\Z)\, \middle \vert \,M J_g M^T = J_g \right \}\,.
\end{equation}
The action of $\spg$ on the period of a genus $g$ surface is 
\begin{equation}
    M \mapsto \widetilde{M} := \Gamma_g \cdot M = (AM + B)(CM + D)
^{-1}, \ \  \forall\; \Gamma_g = \env{pmatrix}{A & B \\ C & D} \in \spg\,.
\label{eq:period_transf}
\end{equation}
The quotient of this action by the Siegel upper half space is called the Siegel fundamental domain $\mathcal{A}_g = \Gamma_g\setminus \mathbb{H}_g$. 

For later use, we will determine the variation of $\widetilde{M}$ from~\cref{eq:period_transf} with respect to $M$. Since $M = M^T$, we have
\begin{equation}
    \partial_{M_{km}} M_{ij} = \frac12\bigg(\delta_{ik}\delta_{jm} + \delta_{kj}\delta_{mi}\bigg)\,.
\end{equation}
Therefore
\begin{equation}
    \begin{aligned}
        \partial_{M_{km}}\widetilde{M}_{ij} &= \partial_{M_{km}}\bigg[ (AM+B)_{in}(CM+D)_{nj} \bigg]\\
        &= \frac12 \bigg((CM+D)^{-1}_{mj}(CM+D)^{-T}_{ik}+ (CM+D)^{-1}_{kj}(CM+D)^{-T}_{im}\bigg)\,.
    \end{aligned}
\label{eq:tilde_var}
\end{equation}

\subsection{Functions, forms, and sections}
We now turn to various objects defined on the Siegel fundamental domain and built from the period matrix $M$. First, we consider the extremely simple combination
\begin{equation}
    M - M^\dag =  M - M^*\,,
\end{equation}
where equality arises from the symmetric nature of the period matrix. Under $\spz$, this combination transforms as 
\begin{equation}
  \begin{aligned}
      \Gamma\cdot M - (\Gamma\cdot M)^\dag&= (A M+B)(C M+D)^{-1} - ( M^\dag C^T+D^T)^{-1}( M^\dag A^T+B^T)\\
      &= ( M^\dag C^T +D^T)^{-1}\bigg\{( M^\dag C^T+D^T)(A M+B)\\ & \quad\; -( M^\dag A^T + B^T)(C M+D)\bigg\} (C M+D)^{-1}\\
      &= (C M^* +D)^{-T} \{ M- M^\dag\} (C M+D)^{-1}\\
      &= (C M+D^{-T})\{ M- M^\dag\}(C M^*+D)^{-1}\,.
  \end{aligned}  
\label{eq:im_transform}
\end{equation}
In simplifying from the second to third lines, we have used the defining properties of elements of $ M\in\spz$: namely, $B^TD - D^TB = 0$, $C^TA-A^TC = 0$, $D^TA-B^TC = 1_4$, and $C^TB-A^TD = -1_4$. In goin to the last line, we used the fact that $\Gamma\cdot M$ is still symmetric. Another function of interest in the main text is
\begin{equation}
    Y( M, M^\dag) := -\frac14 \text{det}( M- M^\dag)\,.
\label{eq:deteq}
\end{equation}
From~\cref{eq:im_transform}, we immediately see that under $\spz$ transformations,
\begin{equation}
    Y(\Gamma\cdot M,\Gamma\cdot M^\dag) =\text{det}(C M^\dag+D)^{-1}Y( M, M^\dag)\text{det}(C M+D)^{-1}\,.
\label{eq:det_transform}
\end{equation}
Less trivially, we will need the transformation law of the derivative of $Y$. This can be determined through the use of~\cref{eq:tilde_var} and~\cref{eq:det_transform}:
\begin{equation}
    \begin{aligned}
        \partial_{ M_{ij}} Y(\Gamma\cdot  M, \Gamma\cdot M^\dag) &= \partial_{ M_{ij}}\Big(\text{det}(C M^\dag+D)^{-1}\text{det}( M- M^\dag)\text{det}(C M+D)^{-1}\Big)\end{aligned}\,,\end{equation}
\begin{equation}
        \begin{aligned}
        \frac{\partial\widetilde{ M}_{ab}}{\partial M_{ij}}\Big(\partial_{\widetilde{ M}_{ab}}Y(\widetilde{ M},\widetilde{ M}^{\dag})\Big) &= \text{det}(C M^\dag+D)^{-1}\Big(\text{det}(C M+D)^{-1} \partial_{ M_{ij}} Y( M, M^\dag)\\&\quad\, + Y( M, M^\dag)\partial_{ M_{ij}}\text{det}(C M+D)^{-1}\Big)\\
        &= \text{det}(C M^\dag+D)^{-1}\text{det}(C M+D)^{-1}\bigg\{ \partial_{ M_{ij}} Y( M, M^\dag)\\&\quad - \frac{Y( M, M^\dag)}{2}\Big(
        C_{ki} (C M+D)^{-1}_{jk}+C_{kj}(C M+D)^{-1}_{ik}\Big)\bigg\}\,,
    \end{aligned}
\label{eq:y_transform}
\end{equation}
where $\widetilde{ M} := \Gamma\cdot  M$ and we have used
\begin{equation}
    \begin{aligned}
\partial_{ M_{ij}}\text{det}(C M+D)^{-1} 
&=\frac{-1}{\text{det}(C M+D)^2}\partial_{ M_{ij}} \text{det}(C M+D)      \\
&=\frac{-1}{\text{det}(C M+D)}\frac12 \Big\{ C_{ki}(C M+D)^{-1}_{jk}+ C_{kj}(C M+D)^{-1}_{ik}\Big\}\,.
    \end{aligned}
\end{equation}
Using~\cref{eq:tilde_var}, ~\cref{eq:y_transform} gives the matrix transformation equation
\begin{equation}
    \partial_{ M}\ln(Y) = \frac{1}{Y} \partial_{ M} Y \rightarrow \frac{1}{Y}(C M+D)\partial_{ M}Y(C M+D)^T - (C M+D)C^T \,.
\end{equation}
We now turn to more complicated functions on the Siegel fundamental domain. Let $F: \mathbb{A}_g \rightarrow \C$ 
be a holomorphic function that satisfies
\begin{equation}
     F_k(\Gamma_g \cdot  M) = v(\Gamma_g) \text{det}(C M + D)^k F_k( M)\,,
\label{eq:classicalSiegel}
\end{equation}
for all $\Gamma_g\in\spg$.~\cref{eq:classicalSiegel} defines a transformation $F_k( M) \mapsto F_k(\Gamma_g \cdot  M)$ under $\spg$. We call any holomorphic function that respects~\cref{eq:classicalSiegel} a
\textit{classical} Siegel modular form of degree $g$, weight $k$, and multiplier system $v(\Gamma_g)$. The space of all classical modular forms of degree $n$ of weight $k$ is denoted by $M_k^{(n)}$. The multiplier systems of $\text{Sp}(2g,\mathbb{Z})$ with $g>1$ are far more restricted than the $\text{SL}(2,\mathbb{Z})$ case of $g=1$. In particular, for $g>1$, it was shown that the full Siegel modular group with $g>1$ admits multiplier systems only for integral $k$ and with values $v(\Gamma_g) = \pm1$~\cite{maass1964multiplikatorsysteme}. See also ~\cite{freitag2024multiplier}.

We  also need to understand how the derivative of a classical Siegel modular form transforms under $\spz$. Of course it will not be a  modular form, but it can be used to construct one. Similar to the discussion above, we have
\begin{equation}
    \begin{aligned}
        \partial_{ M_{ij}} F(\Gamma\cdot M) &= \partial_{ M_{ij}}\Big( \text{det}(C M + D)^k F( M)\Big)\,,\end{aligned}\end{equation}
\begin{equation}
\begin{aligned}
        \frac{\partial \widetilde{ M}_{ab}}{\partial M_{ij}}\partial_{\widetilde{ M}_{ab}}F(\widetilde{ M}) &= \text{det}(C M+D)^k\Big[
        \partial_{ M_{ij}}F( M) \\&\quad+\frac{k}{2}F( M) \left(C^T_{ik}(C M+D)^{-T}_{kj} + (C M+D)^{-1}_{ik}C_{kj}\right)\Big]\,.
    \end{aligned}
\end{equation}
Using~\cref{eq:tilde_var}, this takes the matrix form
\begin{equation}
    \partial_{ M} F( M)\Rightarrow \text{det}(C\tau+D)^k\Big[]
    (C\tau+D)(\partial_{ M} F( M))(C\tau+D)^T +kF( M)(C\tau+D)C^T
    \Big]\,.
\end{equation}
We see indeed that the derivative of a classical Siegel modular form does not transform as a modular form.

In the main text we are also interested in objects that transform differently than \cref{eq:classicalSiegel}. In particular, let $\mathbf{F}( M)$ be a $g\times g$ matrix of functions that satisfies
\begin{equation}
    \mathbf{F}(\Gamma_g\cdot  M) = v(\Gamma_g)(C M+D)\mathbf{F}( M)(C M+D)^T 
\label{eq:vectorSiegel}
\end{equation}
for all $\Gamma_g\in \spg$. Here $v(\Gamma_g)=\pm 1$ is again the multiplier system. We can also vectorize~\cref{eq:vectorSiegel} as
\begin{equation}
    \begin{aligned}
            \text{vec}\left[\mathbf{F}(\Gamma_g\cdot M)\right] &= v(\Gamma_g)\text{vec}\!\left[(C M+D)\mathbf{F}(M)(C M+D)^T\right]\\
            &= v(\Gamma_g)\left((C M +D) \otimes (C M +D)\right)\text{vec}[\mathbf{F}( M)]\,.
    \end{aligned}
\label{eq:vectorizedSiegel}
\end{equation}
From~\cref{eq:vectorizedSiegel}, it is clear that $\mathbf{F}( M)$ is an example of a \textit{vector-valued} Siegel modular form. However, the functions here have a particular transformation matrix formed by the tensor product of $C M+D$ with itself. For clarity, we will define vector-valued Siegel modular forms that transforms specifically with~\cref{eq:vectorSiegel} as \textit{bipartite} Siegel modular forms. 

Let us give an essential example of a bipartite Siegel modular form. If $F( M)$ is a Siegel modular function so that
\begin{equation}
    F(\Gamma_g\cdot  M) = F( M)\,,
\end{equation}
then the gradient matrix $\partial$ is a bipartite Siegel modular form. The proof of this follows from a straightforward calculation:
\begin{equation}
    \begin{aligned}
        \partial_{ M_{ij}} F( M) &=    \partial_{ M_{ij}} F(\Gamma_g\cdot M)= \frac{\partial\widetilde{ M}_{mn}}{\partial  M_{ij}}\frac{\partial}{\partial\widetilde{ M}_{mn}}F(\widetilde{ M})\\
            &= (C M+D)^{-1}_{in}\left(\partial_{\widetilde{ M}_{nm}}F(\widetilde{ M})\right)(C M+D)^{-T}_{mj}\,,
    \end{aligned}
\end{equation}
which we re-arrange to yield the matrix equation
\begin{equation}
    \partial_{ M}F( M)\; \Rightarrow\; (C M+D)(\partial_{ M}F( M))(C M+D)^T\,.
\end{equation}

\subsection{Siegel covariant derivative}
For the physics application in the main text, it is essential to develop a notion of differentiation for SMFs. As mentioned in the previous subsection, generally the derivative of a modular form is not a modular form. However, it is possible to remedy this by defining a covariant derivative.

In theories with $\slz$ modular symmetries, one has to introduce the notion of a quasi-modular form, whose ring is closed under derivation~\cite{quasimodular}.  For every quasi-modular form  there exists a unique almost holomorphic modular form which  modular-completes it. A classic example is the weight-$2$ Eisenstein series $E_2$, which is quasi-modular (of depth $1$) and its almost holomorphic completion is
\begin{equation}
\widehat E_2(\tau)=E_2(\tau)-\frac{3}{\pi}\,(\Im\tau)^{-1}\,.
\end{equation}
One can relate $E_2$ to the modular discriminant $\Delta=\eta^{24}$ via
\begin{equation}
E_2(\tau)=\partial_\tau\log\Delta(\tau)=24\,\partial_\tau\log\eta(\tau)\,.
\end{equation}
For the genus-2 case, the elements of this derivative  are present already in the previous subsection. Let $\chi$ be a classical Siegel modular form of weight $k$, and let $Y$ be the determinant built from the period matrix as defined in~\cref{eq:deteq}. The matrix-valued covariant derivative of $\chi$ is
\begin{equation}
    \nabla_{ M}\chi := \partial_{ M}\chi + k\chi\partial_{ M}\ln(Y) \,.
\end{equation}
This combination will transform as a bipartite modular form with a factor of $\text{det}(C M+D)^k$. That is, the combination
\begin{equation}
    \frac{\nabla_{ M}\chi}{\chi}
\end{equation}
is a bipartite Siegel modular form. Each entry of this matrix is then almost meromorphic, namely it is a finite polynomial in the entries of $Y^{-1}$ and the meromorphy comes from the factor $\chi^{-1}$ (see also~\cite{Klemm:2015iya}).

\subsection{Siegel cusp forms}
Returning to classical Siegel modular forms, we  define the Siegel operator $\Phi: M_k^{(n)} \to M_k^{(n-1)}$ as the mapping between classical modular forms of degree $n$ of (level $1$ and)   weight $k$ to modular forms of degree $n-1$ of (level $1$ and) weight $k$:
\begin{equation}
    \Phi F_k=\lim_{u\rightarrow\infty} F_k\begin{pmatrix}
        T & 0\\0& iu
    \end{pmatrix}\,.
\end{equation}
The elements of ker$(\Phi)$ are called \emph{cusp forms}.

Let us specialize to the Siegel modular threefold $\mathcal{A}_2$. A classical theorem by Igusa \cite{Igusa1962} states that the ring of algebraically independent SMFs over $\mathbb{C}$ is generated by $\sef$, $\ses$, $\chi_{10}$ and $\chi_{12}$.\footnote{To be precise, Igusa later~\cite{igusamodforms}  showed that for the full ring of $\spz$ modular forms, one needs an additional generator $\chi_{35}$. The ring then reads 
\begin{equation}
    M(\spz)
\;\cong\;
\mathbb{C}[\sef,\ses,\chi_{10},\chi_{12},\chi_{35}]/(R(\sef,\ses,\chi_{10},\chi_{12})-\chi_{35}^2)\,,
\end{equation}
with $R$ and explicit isobaric polynomial in $\sef$, $\ses$, $\chi_{10}$ and $\chi_{12}$. 
Hence, $\chi_{35}$ is algebraically dependent on the other elements.} 
The weight 4 and 6 SMFs are the degree 2 Eisenstein series defined as 
\begin{equation}
\mathcal{E}_k( M)=\sum_{C,D} \det (C M+D)^{-k}\,,
\end{equation}
where the summation is over all inequivalent bottom rows $(C\;\;D)$ of elements of $\Gamma_2$.  The latter two modular forms are the cusp forms, normalized as \cite{Igusa1962}
\begin{equation}\label{eq:cusps}
    \begin{split}
        \chi_{10}:=& -\frac{43867}{371\cdot 2^{12}\cdot 3^5\cdot 5^2}\left(\sef\ses-\mathcal{E}_{10}\right)\,, \\
        \chi_{12}:=&\frac{131\cdot 593}{2^{13}\cdot 3^7\cdot 5^3\cdot 7^2\cdot 337}\left(3^2\cdot7^2\sef^3+2\cdot5^3\ses^2-691\mathcal{E}_{12}\right)\,.
    \end{split}
\end{equation}
A modular function is a modular form of weight $0$.

A SMF $F_k$ of weight $k$ has a Fourier-Jacobi expansion~\cite{eichlerzagier}
\begin{equation}
F_k( M)=\sum_{m\geq 0}\varphi_{k,m}(T,Z)p^m\,,\quad p=e^{2\pi i  U} \,,
\end{equation}
where the coefficients $\varphi_{k,m}$ of the expansion are Jacobi forms of weight $k$ and index $m$,
\begin{equation}
    \varphi_{k,m}(T,Z)=\sum_{l,b}c(l,b)q^{l}y^{b}\,, \qquad q=e^{2\pi i T}\,,\quad y=e^{2\pi i Z}\,.
\end{equation}
If we define
\begin{equation}
    A= \frac{1}{24}\sum_l c(0,l)\,,\quad B=\frac{1}{2}\sum_{l>0}lc(0,l)\,,\quad C=\frac{1}{4}\sum_{l}l^2 c(0,l)\,,
\end{equation}
then the exponential lift of the nearly holomorphic Jacobi form $\varphi_{0,t}$ is the product
\begin{equation}
    \text{Exp-Lift}(\varphi)( M)=q^Ay^Bp^C\prod_{n,m,l>0}\left(1-q^ny^lp^{tm}\right)^{c(nm,l)}\,,
\end{equation}
which defines a meromorphic modular form of weight $c(0, 0)/2$ with respect to the subgroup $\Gamma_t$ of $\spz$. 
If $\varphi$ is a \emph{weak Jacobi form}, then we
actually have $C= tA$, and $\sum_{l}c(n,l)=0 \;\;\forall \;n>0$. The zeroes or poles of meromorphic SMFs that can be cast as exponential lifts can be identified by choosing $T,\,U,\,Z$ such that $q^ny^lp^{tm}=1$ in one of the factors: this will make that factor vanish, so that the product either vanishes or diverges~\cite{gritsenko1998automorphic}. For the results stated in the main text, we should look at the behavior near $Z=0$: the leading zero or pole near $Z = 0$, up to numerical coefficients, is then
\begin{equation}
    q^Ap^{tA}\prod_{m>0}\left(1-p^{tm}\right)^{24 A}\prod_{n>0}\left(1-q^{n}\right)^{24 A}\prod_{l}\left(1-y^{l}\right)^{c(0,l)}\sim Z^{\sum_{l}c(0,l)}\eta(T)^{24 A}\eta(tU)^{24 A}
\end{equation}
with $l<0$ (meaning that we are restricting on the the dimension-2 locus where $Z=0$.) The asymptotic behavior of the cusp forms for $Z\rightarrow 0$ is~\cite{Igusa1962}
\begin{align}
    \chi_{10} &=Z^2 \eta^{24}(T) \eta^{24}(U) + \mc{O}(Z^4)\,,\\
\chi_{12} &= \eta^{24}(T) \eta^{24}(U)+ \mc{O}(Z^2)\,.
\end{align}

\subsection{Generators of Stab$(\sigma_i)$}\label{app:generators}
The generators mentioned in \cref{tab:fixedpts} are given by:
    \begin{align}
        h^{(1)}_1 = \begin{pmatrix}
                        0 & -1 & -1 & -1\\
                        0 & 0 & -1 & 0\\
                        0 & 0 & 0 & -1\\
                        1 & 0 & 0 & 1
                    \end{pmatrix}\coma
      h^{(2)}_1=h^{(3)}_3=h^{(4)}_2 = \begin{pmatrix}
                        0 & 1 & 0 & 0\\
                        1 & 0 & 0 & 0\\
                        0 & 0 & 0 & 1\\
                        0 & 0 & 1 & 0
                    \end{pmatrix}\coma
    h_2^{(6)} = \begin{pmatrix}
                        -1 & 0 & 0 & 0\\
                        0 & -1 & 0 & 0\\
                        0 & 0 & -1 & 0\\
                        0 & 0 & 0 & -1
                    \end{pmatrix}
    \end{align}
   \begin{align}
       h^{(2)}_2 = \begin{pmatrix}
                        -1 & 1 & 1 & 0\\
                        1 & 0 & 0 & 1\\
                        -1 & 0 & 0 & 0\\
                        1 & -1 & 0 & 1
                    \end{pmatrix}\coma
         h^{(3)}_1 = \begin{pmatrix}
                        0 & 0 & 1 & 0\\
                        0 & -1 & 0 & 0\\
                        -1 & 0 & 0 & 0\\
                        0 & 0 & 0 & -1
                    \end{pmatrix}\coma
        h^{(3)}_2 = \begin{pmatrix}
                        0 & 0 & -1 & 0\\
                        0 & 0 & 0 & 1\\
                        1 & 0 & 0 & 0\\
                        0 & -1 & 0 & 0
                    \end{pmatrix},\,
   \end{align}
        \begin{align}
        h^{(4)}_1 = \begin{pmatrix}
                        0 & 0 & 0 & -1\\
                        1 & 0 & 1 & 0\\
                        0 & 1 & 0 & 1\\
                        -1 & 0 & 0 & 0
                    \end{pmatrix}\coma       
        h^{(4)}_3 = \begin{pmatrix}
                        0 & 0 & 1 & 0\\
                        0 & 0 & 0 & -1\\
                        -1 & 0 & -1 & 0\\
                        0 & 1 & 0 & 1
                    \end{pmatrix}\coma
        h^{(5)}_1 = \begin{pmatrix}
                        0 & 0 & 0 & 1\\
                        0 & 0 & 1 & 1\\
                        1 & -1 & 0 & 0\\
                        -1 & 0 & 0 & 0
                    \end{pmatrix},\,
    \end{align}
     \begin{align}
       h^{(5)}_2 = \begin{pmatrix}
                        0 & 0 & 1 & 1\\
                        0 & 0 & 1 & 0\\
                        0 & -1 & 0 & 0\\
                        -1 & 1 & 0 & 0
                    \end{pmatrix}\coma
        h^{(5)}_3 = \begin{pmatrix}
                        0 & 0 & 0 & 1\\
                        0 & 0 & -1 & 0\\
                        0 & -1 & 0 & 0\\
                        1 & 0 & 0 & 0
                    \end{pmatrix}\coma 
                     M_4= h^{(6)}_1 = \begin{pmatrix}
                        0 & 0 & 1 & 0\\
                        0 & 0 & 0 & 1\\
                        -1 & 0 & -1 & 0\\
                        0 & -1 & 0 & 0
                    \end{pmatrix}\fstop
 \end{align}

In the main text, we will be interested in eigenvalues in matrices constructed from the above generators. To be more precise, writing the stabilizer matrices as in~\cref{eq:stabmatrix}:
\begin{equation}
h_a^{(i)}=\begin{pmatrix}
    A_a^{(i)}& B_a^{(i)}\\
    C_a^{(i)}& D_a^{(i)}\,.
\end{pmatrix}
\end{equation}
We will be interested in eigenvalues of the quantities
\begin{equation}
    N^{(i)}_a := C_a^{i}\sigma_i + D_a^{(i)}\,,
\label{eq:transform_matrix}
\end{equation}
where $\sigma_i$ is the $i$-th fixed point according the labels in~\cref{tab:fixedpts}.  These are of course the matrices that enter in the transformation laws of both classical Siegel modular forms and, more importantly, the bipartite Siegel modular forms.

\begin{table}
    \centering
    \begin{tabular}{|c|c|}
    \hline
        Matrix & Eigenvalues \\
         \hline
         $N^{(1)}_1$ & $\exp[4i\pi/5]$, $\exp[3i\pi/5]$ \\
        \hline
         $N_1^{(2)}$ & $-1,1$ \\
         $N_2^{(2)}$ & $\exp[i\pi/4]$, $\exp[-i\pi/4]$\\
         \hline
         $N_1^{(3)}$  & $-1,-i$ \\
         $N_2^{(3)}$ & $i,-i$ \\
         $N_3^{(3)}$ & $-1,1$ \\
         \hline
         $N_1^{(4)}$ & $-\sqrt{-\omega(1+\omega)},\sqrt{-\omega(1+\omega)}$ \\
         $N_2^{(4)}$ & -1,1 \\
         $N_3^{(4)}$ & $-(1+\omega)$, $(1+\omega)$\\
         \hline
         $N_1^{(5)}$ & $-i$, $i$ \\
         $N_2^{(5)}$ & $-i$, $i$ \\
         $N_3^{(5)}$ & $-1$, $1$ \\
         \hline
         $N_1^{(6)}$ &  $-i$, $-(1+\omega)$\\
     \hline
    \end{tabular}
    \caption{Eigenvalues of the matrices $N^{(i)}_a$ defined by the fixed points and their stabilizer group generators as defined in~\cref{eq:transform_matrix}. Note that we still identify $\omega = \exp[2i\pi/3]$.}
    \label{tab:generator_eigs}
\end{table}

We will also be interested in the eigenvectors of several matrices related to the $\{N^{(i)}_a\}$. We define
\begin{equation}
    M_a^{(i)} = 1_4 - N_a^{(i)}\otimes N_a^{(i)}\,.
\end{equation}
The $\{M^{(i)}_a\}$ are of course the transformation matrices of the bipartite Siegel modular forms. For the matrices associated with the fixed points $\{\sigma_3,\sigma_5\}$, we will need the eigenvectors of the $\{M^{(i)}_a\}$ that have zero eigenvalue. We list these here:
\begin{table}[h!]
    \centering
    \begin{tabular}{|c|c|}
    \hline
        Matrix & Eigenvectors \\
         \hline
         $M^{(3)}_1$ & $(0,0,0,1)^T$ \\
         $M^{(3)}_2$ & $(0,0,1,0)^T$, $(0,1,0,0)^T$\\
         $M^{(3)}_3$ & $(1,0,0,1)^T$, $(0,1,1,0)^T$\\
     \hline
        $M^{(5)}_1$ & $(1,2,0,-2)^T$, $(0,-1,1,0)^T$\\
        $M^{(5)}_2$ & $(-2,2,0,1)^T$, $(0,-1,1,0)^T$\\
        $M^{(5)}_3$ & $(-1,0,0,1)^T$, $(-1,1,1,0)^T$\\
    \hline
    \end{tabular}
    \caption{Eigenvectors (up to normalization) with null eigenvalue for the matrices $M^{(i)}_a$ with $i=3,5$.}
    \label{tab:generator_eigenvects}
\end{table}

\bibliographystyle{JHEP}
\bibliography{MainBib}

\end{document}